\documentclass{article}
\usepackage{hyperref}
\usepackage{xcolor}
\usepackage{mathrsfs}

\usepackage{comment}
\usepackage{natbib}
\usepackage{subcaption}
\usepackage{tikz}
\usetikzlibrary{matrix,arrows}
\usepackage{amsmath,amssymb,amsfonts,mathrsfs} 
\usepackage{enumerate} 
\usepackage[T1]{fontenc}
\usepackage{lmodern}
\usepackage[utf8]{inputenc}
\usepackage[kerning]{microtype} 
\usepackage{booktabs}
\usepackage{amsthm}
\usepackage{mathtools}
\usepackage{float}
\usepackage{mathtools}

\usepackage{setspace}
\newtheorem{theorem}{Theorem}

\newtheorem{definition}[theorem]{Definition}

\newtheorem{lemma}[theorem]{Lemma}

\newtheorem{proposition}[theorem]{Proposition}
\newtheorem{remark}[theorem]{Remark}

\usepackage{bbold}

\usepackage{tocloft}
\usepackage{url}

\cftsetindents{section}{0em}{2em}
\cftsetindents{subsection}{0em}{2em}

\mathtoolsset{showonlyrefs=true}
\title{Stochastic Stability of 
a Recency Weighted Sampling Dynamic\thanks{We thank Michel Benaïm, Linus Bergqvist, Lee Dinetan, Boualem Djehiche, Isak Trygg Kupersmidt, Mark Voonerweld, Jörgen Weibull, Peter Wikman, and seminar participants at SSE, MIT theory lunch, Stockholm PhD Math Fest, and SING15 for helpful comments. In particular, Lee Dinetan was involved in the early stages of this project and we would like to thank him for many helpful insights, not least into questions about ergodicity. The research of A. Aurell was supported by the Swedish Research Council (2016-04086) and AFOSR \# FA9550-19-1-0291; G. Karreskog by Tom Hedelius and Jan Wallander Foundation, and Knut and Alice Wallenberg Research Foundation.}}
\author{Alexander Aurell\footnote{Department of Operations Research and Financial Engineering, Princeton University aaurell@princeton.edu} \and Gustav Karreskog\footnote{Department of Economics, Stockholm School of Economics \mbox{gustav.karreskog@phdstudent.hhs.se}}}

\date{}
\begin{document}

\maketitle

\begin{abstract}
We introduce and study a model of long-run convention formation for rare interactions. Players in this model form beliefs by observing a recency-weighted sample of past interactions, to which they noisily best respond. We propose a continuous state Markov model, well-suited for our setting, and develop a methodology that is relevant for a larger class of similar learning models. We show that the model admits a unique asymptotic distribution which concentrates its mass on some minimal CURB block configuration. In contrast to existing literature of long-run convention formation, we focus on behavior inside minimal CURB blocks and provide conditions for convergence to (approximate) mixed equilibria conventions inside minimal CURB blocks. 
\medskip \\ \medskip
\textbf{JEL:} C72, C73
\\
\textbf{Keywords:} Evolutionary game theory, learning in games, stochastic stability, recency, mixed Nash equilibria, minimal CURB blocks
\end{abstract}

\newpage
\tableofcontents
\section{Introduction}

Social conventions form and evolve in many real life situations. In this paper we consider the formation of conventions as a repetitive feedback process where expectations of behavior are formed by observing other interactions in a society, expectations then inform decisions which alters the history from which other members of society form expectations. For example, when buying a house each bidder (player) might not have participated in the exact same bidding (game) before, but has knowledge about some, but not all, previous interactions and assumes that the other bidders interacting with her will behave similarly to how bidders have historically behaved. 

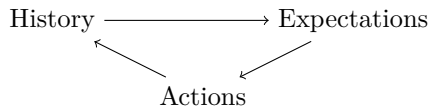
\begin{figure}[H]
\centering
\begin{tikzpicture}
    \node (1) at (0,0) {History};
    \node (2) at (4,0)  {Expectations};
    \node (3) at (2, -1) {Actions};
    \draw[->] (1) edge (2) (2) edge  (3) (3) edge (1);
\end{tikzpicture}
\caption{Two players are randomly selected from large populations and assigned roles. The players form expectations by sampling from historical records of interactions and then act based on those expectations. The realized play is appended to the history, and the process is repeated.}
\label{fig:diag}
\end{figure}

More specifically, the general setting in this paper is the evolution of social conventions as introduced by \cite{Young1993}. We imagine that for each player role there is a large population of candidates from which players are randomly drawn to play a normal form game. The populations are large in the sense that the same player is never selected twice.\footnote{For example, we can imagine that there is a continuum of players in each population and that the random draw is done by sampling from an atomless probability measure over the continuum.} Before deciding which action to take the players access a sample of historical interactions. The players use the sample to form beliefs about the opposite roles' historical behavior (this is the only information the players can base their decision on, since they have never played the game before). Thereafter, the players simultaneously respond to the mixed strategy induced by the sample. Once they have played, their interactions is appended to the history, new players are randomly drawn from the populations, and the process is repeated with the updated history. 

By modeling repeated play based on historical records as diagrammed in Figure~\ref{fig:diag}, one hopes to answer questions about which actions will be taken in the long run, and therefore which stable conventions, if any, will arise. We will refer to a dynamical model for the likelihood of the interactions as a \textit{learning process}. We interpret stable points of this learning process as stable social conventions. This setting differs from that of a large population repeatedly playing the same game, examples of which can be found in for example \cite{sandholm2010population}. In such models, each player plays the same game many times, and over time learn which actions to take. In our setting, any given player only plays the game once, but has partial information about how other players typically play the game. It is thus suitable for studying the formation of conventions in strategic interactions where any given individual only rarely participates. 

When studying the asymptotic distribution of the (state of the) learning process it is convenient, both theoretically and numerically, if it is an ergodic Markov process. That is achieved in the original formulation of \cite{Young1993} by defining the state of the learning process as a finite sequence, the "finite memory" containing the most recent interactions, by letting the players form beliefs by sampling strategies from the memory without replacement, and by assuming a small mistake probability with which a random action is taken instead of a best reply. The finite state space and the noisy action (the possibility of making a mistake) ensures that Young's learning process has a unique invariant distribution to which it converges asymptotically.

Much of the work building Young's original model contains the finite memory and noisy action structure, which is well suited for studying the relative stability of different pure (i.e., strict)
Nash equilibria or minimal CURB blocks.\footnote{A subset (block) of strategy profiles $C$ is called Closed Under Rational Behavior (CURB) if the best replies to any strategy profile with support in $C$ is also in $C$. It is called a minimal CURB block if it does not contain any strictly smaller CURB block \cite{Basu1991}.}
However, finite memory based learning is ill-suited to answer questions about the players' behavior around mixed Nash equilibria since its evolution depends heavily on the ordering of the history, not just the sampling probabilities. Furthermore, it exhibits behavior around even simple mixed Nash equilibria that is better viewed as a  modeling artifact than as a realistic description of behavior. The purpose of this paper is to define a new learning process with the following features: firstly, it converges to a minimal CURB configuration and secondly, it behaves reasonably also inside non-singleton minimal CURB containing one unique mixed Nash equilibrium.

To address the problem of potentially unwanted cycling and increase the stability of social conventions we introduce the \textit{Recency Weighted Sampler} (RWS). It is a learning process that differs from previous work in its structure of the historical record of plays. The history is assumed to be infinite with recent interactions more likely to be sampled. A finite sample is drawn with replacement by each player at each period. The probability of sampling the interaction of a past game decreases by a factor $\beta$, $0 < \beta < 1$, per game that has been played since. This geometric decrease allows us to use the sampling probabilities, of strategies, as the state space of the learning process. The Markovian property of the process is preserved and we can in a meaningful way analyze it at a finer level inside 
minimal CURB blocks (and determine properties of the distribution of interactions, i.e., the social convention, inside a minimal CURB block).

\subsection{Related Literature}   

Already in his dissertation John Nash gave a second interpretation of the Nash equilibrium, the \emph{mass action} interpretation \citep{Nash1950}. He assumes that a large population is associated to each player role, that one player per role is selected in each period to play the game, and that the individual players accumulate empirical information on the relative advantage of the different available pure strategies. He then argues, informally, that in such a setting, the stable points correspond to Nash equilibria and those points should 
eventually be reached by the process. 

The mass action interpretation is appealing since its assumptions about bounded rationality and repeated interactions are more credible than those underlying the rationalistic interpretation built on assumptions of perfect rationality and common knowledge.\footnote{Especially since perfect rationality and common knowledge by itself only leads to rationalizability but not all the way to Nash equilibrium.} Furthermore, experimental evidence often favors some kind of learning and adjustment over the rationalistic motivation. The general result is that in a one-shot interaction, play rarely corresponds to a Nash equilibrium, but if the players have a chance to learn and adjust, 
play often (but far from always) moves to a Nash equilibrium. See \cite[Ch. 6]{Camerer2003} for an overview of experimental models and results. 

Appealing as the motivation might be, the theoretical picture has turned out to be considerably more complicated than indicated by Nash's informal argument. One of the first, and most studied, models formalizing a setting similar in spirit to the mass action interpretation is that of fictitious play in \cite{brown1951iterative}.
Even though Brown thought fictitious play would in general converge to a Nash equilibrium, it was shown in \cite{shapley1964some} that even in a game with a unique Nash equilibrium there might only exist a stable cycle and no convergence to the mixed equilibrium. In general, it is the case that if the process has a stationary point, it must be a Nash equilibrium, but the existence of such a stationary point is not guaranteed. See e.g. \cite{fudenberg1998theory, weibull1997evolutionary, sandholm2010population} for overview of such results. Existing general results do not address convergence to stable points (which normally correspond to Nash equilibria) but convergence to stable sets. \cite{Ritzberger1995} show set-convergence results for evolutionary dynamics and \cite{Balkenborg2013} for best reply dynamics. Similarly \cite{Hurkens1995, young1998individual} show set-convergence results for dynamics similar to those studied in this paper. 

Smooth fictitious play, first introduced in \cite{Fudenberg1993}, 
is a variant of fictitious play where players respond with a perturbed best response. In contrast to the standard version of fictitious play, 
not only the empirical frequency but also actual play can converge to 
a Nash equilibrium. In \cite{Benam1999, Hofbauer2002}, global convergence results are shown for some games with unique Nash equilibria, including interior ESS, two-player zero-sum, supermodular, and potential games. 

A downside with standard versions of fictitious play and smooth fictitious play is that the increments of the learning processes decrease in size over time. In practice the point of initialization is therefore crucial for convergence. Furthermore, if the behavior is cyclic the cycles take longer and longer to complete. Introducing a bias towards more recent plays, similar to that used in this paper, yields processes with increments of similar size over time, which for many applications is natural. Such processes are studied in \cite{Benaim2009}, where the time average in unstable games is studied, and in \cite{fudenberg2014learning}. 

The one class of dynamics for which we have quite general results for convergence to equilibrium rely on a combination of noisy behavior and satisfaction \citep{Foster2003,PeytonYoung2006,Hart2006,Block2019learning}. A given player randomly explores the action space until she is satisfied, e.g., her received payoff is higher than some threshold 
or close enough to the maximum payoff observed. Then she keeps taking that action as long as she still is satisfied. The exact setting and formulation of results vary, but in general models in this category are able to converge to a Nash equilibrium under general circumstances. A possible downside is that the path to equilibrium can be very long and somewhat unrealistic. The players are thus, in a sense, too unsophisticated, at least if they have full knowledge about the game.

The existing literature building on \cite{Young1993, young1998individual} has not focused on the convergence to mixed Nash equilibria, but instead on the speed of convergence \cite{Kreindler2013} or improving tools for finding stochastically stable subsets \cite{Ellison2000}. To the best of our knowledge no one has conducted a careful study of the convergence for this type of learning processes to mixed Nash equilibria. 

\subsection{Summary and Outline}

In Section~\ref{sec:model} the proposed learning process, the Recency Weighted Sampler, is denfied and we introduce the probabilistic and analytical tools needed to analyze the process. Since we define a framework different from existing models (most crucially, RWS has a continuous state space) we cannot rely directly on any existing results. We therefore begin by proving some standard properties of the learning process in Section \ref{sec:mainresults}. We prove uniform ergodicity for a class of learning process of which the RWS is a member. Next, we show that in the small-error limit the invariant distribution of the RWS will concentrate on minimal CURB blocks. Once we have recovered these properties, we analyze RWS's behavior inside minimal CURB blocks that are non-singleton, and show that for any generic game where the minimal CURB blocks are at most $2\times 2$ play will eventually concentrate around Nash equilibria or, when the sample size is small, around points close to the Nash equilibria. The paper concludes with Section~\ref{sec:concl} where we discuss the results and possible extensions. Proofs have been appended in the end of the paper.

\section{The Recency Weighted Sampler}
\label{sec:model}

Let $G$ be a finite two-player game, iteratively played by two players drawn from large populations. We assume that a player never plays the game more than once. The game has two asymmetric player roles, $1$ and $2$. The sets of pure strategies in the game are $S_{1}$ and $S_{2}$, containing $m_{1}\in \mathbb{N}$ and $m_{2}\in \mathbb{N}$ pure
strategies respectively; the spaces of mixed strategies are thus $\Delta\left( S_{1}\right) $ and $\Delta \left( S_{2}\right) $. Throughout the paper, $-i$ denotes the index $\{1,2\}\backslash\{i\}$, $i\in\{1,2\}$. For $\sigma \in \Delta \left( S_{-i}\right) $, we denote by $BR_{i}\left(\sigma \right) \subset S_{i}$ the set of best replies of player $i$ to the mixed strategy $\sigma$. We identify $\Delta \left( S_{i}\right) $ with the $\left( m_{i}-1\right) $-dimensional
simplex and denote $\square \left(S\right) :=\Delta \left( S_{1}\right) \times \Delta\left( S_{2}\right)$, $\square(S)$ being endowed with the Euclidean distance $\| \cdot \| $. We denote by $\mathcal{B}(\square(S))$ and $\mathcal{P}(\square(S))$ the Borel $\sigma$-field over $\square(S)$ and the set of Borel probability measures over $\square(S)$, respectively.

\subsection{The Stochastic Best Reply in the RWS}

Each interaction is recorded as a strategy pair $\left( s_{1},s_{2}\right) $. Denoting $s_{1}\left( t\right) $ and $s_{2}\left( t\right) $ the actions taken at time $t \in \mathbb{Z}$, for player 1 and player 2 respectively. The history is thus a sequence of plays
\begin{equation}
    \left( \left( s_{1}\left( t\right) ,s_{2}\left( t\right) \right) \right)_{t\in \mathbb{Z}}.
\end{equation}
The purpose of extending the history to $t<0$ is purely technical. We note here that such an artificial "initial history" represents an initial sampling distribution at $t=0$ over the action space and later we will see that any such distribution can be represented by an infinite history of interactions. If the payers initially have no knowledge, this can be taken as a uniform distribution over the strategy space.  

At each time $t$, each player of role $i\in \left\{ 1,2\right\} $ samples $k\in \mathbb{N}$ plays (with replacement) from the history of
the opposing player role $-i$. Each sample is drawn independently 
and samples are drawn with bias towards more recent plays in a geometric fashion. Namely, players of role $i$ have a bias $\beta
\in \left( 0,1\right)$, called the \textit{recency parameter}, such that at time $t$ the probability of sampling the interaction from time
period $t-\tau$, $\tau \in \{1,2,\dots\}$ is
\begin{equation}
    \left( 1-\beta\right) \beta^{\tau-1}.
\end{equation}
Therefore, a play of the strategy $s \in S_{-i}$ 
will be sampled by player $i$ with probability
\begin{equation}
    p_{-i,s}\left(t\right) 
    = \left( 1-\beta\right) \sum_{\tau=1}^{\infty}\beta^{\tau-1} 1_s(s_{-i}(t-\tau)),
\end{equation}
where $1_s$ is the indicator function on $s$.

We will call $p_{i}\left( t\right) :=\left( p_{i,1}\left( t\right),\ldots p_{i,m_{i}}\left( t\right) \right) $ the state process of player role $i$ at time $t$ and $p(t) := (p_1(t),p_2(t))$ the state process or the learning process, interchangeably. $p_i(t)$ is a vector of sampling probabilities obtained by player $i$ from player $-i$'s history and is an element of $\Delta \left( S_{-i}\right) $. The result of player $i$'s sampling is a random vector $\left( n_{-i,1}\left( t\right) ,\ldots n_{-i,m_{-i}}\left( t\right) \right) $ of integers,
multinomially distributed with parameters $k$ and $p_{-i}\left( t\right)$. For $s \in S_{i}$, let $\overrightarrow{1_{i,s}}\in \Delta \left(S_{i}\right) $ be the unit vector representing the pure strategy $s \in S_{i} $, i.e., a vector of of size $m_{i}$ with $0$ everywhere except at position $s$, where it is $1$. From her sample, player $i$ forms an empirical (average) opposing strategy profile
\begin{equation}
\label{eq:samples}
    D_{-i}\left( t\right) 
    := \frac{1}{k}\sum_{s=1}^{m_{-i}}n_{-i,s}\left(t\right) \overrightarrow{1_{-i,s}}\in \Delta \left( S_{-i}\right).
\end{equation}

Player $i$ now acts as if her opponent will play according to the
mixed strategy $D_{-i}\left( t\right) $ and tries to play a best response to it. However, she can make a mistake. Player $i$'s \textit{error parameter} (or mistake frequency) $\varepsilon \in \left[ 0,1\right] $ indicates the probability she will fail to play a strategy in $BR_{i}\left( D_{-i}\left( t\right) \right) $, and instead play a strategy in $S_{i}$ at random (with uniform probability). If $BR_{i}\left( D_{-i}\left( t\right) \right) $ is not a singleton, the realized action is sampled uniformly from all the elements of $BR_{i}\left( D_{-i}\left( t\right) \right)$. We denote the outcome of the uniform sampling between all best replies to $\sigma\in \Delta (S_{-i})$ by $\widehat{BR}_i(x) \in S_i$. The distinction we want to emphasize with this notation is that $BR_i(x)$ is set-valued (the set of all best replies to $x$) while $\widehat{BR}_i(x)$ is $S_i$-valued
and random (since a best reply was randomly selected from the set $BR_i(x)$).

Ultimately, player $i$ will play $\widehat{BR}_{i}\left( D_{-i}\left( t\right) \right)$, with $D_{-i}(t)$ obtained as described above, with a probability of $1-\varepsilon $; and additionally, play any strategy $ s \in S_{i}$ with probability $\varepsilon/m_i$. We complete this section by defining the random variable
\begin{equation}
    \widetilde{BR}_{i}\left( p_{-i}\right) \in S_i
\end{equation}
to be the random choice of strategy obtained by a player $i$ through the following process:
\begin{enumerate}
    \item Accessing the history of interactions from which plays by the opposing role are sampled with probabilities given by $p_{-i}$;
    \item Sampling $k$ opponent actions to form the empirical belief $D_{-i} \in \Delta(S_{-i})$;
    \item Playing the best response $\widehat{BR}_{i}\left( D_{-i}\right) $, except in a fraction $\varepsilon $ of the time when a randomly selected strategy is played.
\end{enumerate}

\subsection{The Dynamics of RWS}
\label{sec:run}

At $t=0$, an initial history $\left( \left( s_{1}\left( u\right) ,s_{2}\left(u\right) \right) \right) _{u\in \mathbb{Z}_{-}}$, $s_i(u)\in S_i$, is given. At each time $t\in \mathbb{N}_{0}$, 
two new individuals are assigned to the roles. They use the same parameters values $k$, $\beta$, and $\varepsilon $.\footnote{This is not a necessary assumption but we make it for the sake of presentation. The RWS can be defined with a different set of parameter values for each player role.} After sampling from the history with recency parameter $\beta$, they play $s_{i}\left( t\right) =\widetilde{BR}_{i}\left( p_{-i}\left( t\right) \right) $, $i=1,2$, where $p_{-i}\left( t\right) $ is exactly the historical distribution of plays with recency bias. The realized strategy profile $\left(s_{1}\left( t\right) ,s_{2}\left( t\right) \right) $ is appended to the history, and the procedure repeats. The exponential nature of sampling leads to the following characterization of the RWS learning process. A proof can be found in Appendix~\ref{app:exponential_history}.

\begin{proposition}
\label{exphist}
The state process of player $i$ obeys the equation
\begin{equation}
\label{eq:RWS-dynamics}
    p_{i}\left( t+1\right) 
    = \beta p_{i}\left( t\right) +\left( 1-\beta \right)\overrightarrow{1_{i,s_{i}(t)}}, \quad t \in \mathbb{Z}
\end{equation}
where $s_{i}\left( t\right)  = \widetilde{BR}_{i}\left( p_{-i}\left( t\right)\right) $ is drawn randomly according to the model.
\end{proposition}

The order of historical plays is not necessary to characterize the
model, all the relevant information is captured in $\left( p_{1}\left( t\right) ,p_{2}\left( t\right) \right) \in \square\left( S\right)$. From the position $\left( p_{1}\left( t\right) ,p_{2}\left( t\right) \right) \in \square\left( S\right)$, at most $m_{1}m_{2}$ different points $\left(p_{1}\left( t+1\right) ,p_{2}\left( t+1\right) \right) $ may be reached. Conditioned on $p(t)$, for any $s_{1}\in S_{1}$ and $s_{2}\in S_{2}$ the point
\begin{equation}
    \left( \beta p_{1}\left( t\right) +\left( 1-\beta\right) 
    \overrightarrow{1_{1,s_{1}}},\beta p_{2}\left( t\right) +\left(
    1-\beta\right) \overrightarrow{1_{2,s_{2}}}\right)
\end{equation}
will be reached when $s_{1}\left( t\right) =s_{1}$ and $s_{2}\left( t\right)=s_{2}$, which happens with probability
\begin{equation}
    \prod_{i=1}^2 \mathbb{P}\left(\widetilde{BR}_i(p_{-i}(t)) = s_i \ |\ p_{-i}(t)\right),
\end{equation}
since we assume players sample independently of each other, and
\begin{equation}
\label{eq:the_prob_that_is_Lip}
    \mathbb{P}\left( \widetilde{BR}_{i}\left( p_{-i}\left( t\right) \right) = s_{i}\right) 
    = \left( 1-\varepsilon \right) \mathbb{P}\left( \widehat{BR}_{i}\left(D_{-i}\left( t\right) \right) =s_{i}\right) +\varepsilon/m_i,
\end{equation}
where $D_{-i}\left( t\right) \in \Delta \left(S_{-i}\right)$ is a multinomial combination of strategies (with parameters $k$ and $p_{-i}\left( t\right) $). 

\subsection{Markovianity}
\label{sec:K}

By construction $(p(t); t\in\mathbb{N})$ is a Markov chain taking values in $\square\left( S\right)$. Since its state space is the continuous set $\square(S)$, its transition kernel is a function $P : \square(S) \times \mathcal{B}(\square(S)) \rightarrow \mathbb{R}$ with the standard Markov kernel properties. The kernel takes a tuple $(x,B)$ and returns the probability of the chain transitioning from $x$ to $B$ in one period. The kernel is the continuous state space equivalent of the transition rate matrix in discrete state space models. From the dynamics of RWS, we have that for all $(p_1,p_2)\in\square(S)$ and $B\in\mathcal{B}(\square(S))$,
\begin{equation}
\label{eq:def_of_kernel}
    \begin{aligned}
    P\left(( p_{1},p_{2}), B\right) 
    &= \sum_{s_{1}=1}^{m_{1}}\sum_{s_{2}=1}^{m_{2}}
    \mathbb{P}\left( \widetilde{BR}_{1}\left( p_{2}\right) =s_{1},\widetilde{BR}_{2}\left( p_{1}\right) =s_{2}\right) \times
    \\
    &\qquad \qquad 1_B\left( \beta p_{1}+\left( 1-\beta\right) \overrightarrow{1_{1,s_{1}}},\beta p_{2}+\left( 1-\beta\right) \overrightarrow{1_{2,s_{2}}} \right).
    \end{aligned}
\end{equation}

\begin{remark}
An underlying assumption of the RWS is that there exists a probability space $\left( \Omega ,\mathcal{F},\mathbb{P}\right) $ carrying all the random variables necessary for defining the learning process. The space is filtered by $\mathbb{F}$, the natural filtration of the state process, and satisfies the usual conditions. The assumption is innocent, it only requires the space to carry a countable number of independent random variables. It is in this filtered space that we subsequently study the learning process as a Markov chain.
\end{remark}

\subsection{Convention Formation in Matching Pennies}

To better understand why RWS is suitable for the formation of mixed strategy conventions, consider perhaps the  simplest normal form game with a unique mixed Nash equilibrium: Matching Pennies, presented in Table~\ref{fig:match-pennies-bimatrix}. In this elementary example, we compare the behavior of RWS to the behavior of the finite memory process of \cite{Young1993} by simulation. 
\begin{table}	
\centering
    \begin{tabular}{c|rc|rc|}
    \multicolumn{0}{c}{}&  \multicolumn{2}{c}{$\mathbf{1}$} & \multicolumn{2}{c}{$\mathbf{2}$} 
    \\
    \cline{2-5}
    $\mathbf{1}$ & $1,$ & $-1$ &  $-1,$ &$ 1$ 
    \\ 
    \cline{2-5} 
    $\mathbf{2}$ & $-1,$ & $1$ & $1,$ &  $-1$ 
    \\ 
    \cline{2-5}
    \end{tabular}
\caption{The Matching Pennies payoff bimatrix. The row player has the "agreeing" role, aiming to match strategy with the column player, who has the "disagreeing" role, and aims to play differently than the row player. The unique mixed Nash equilibrium is $\left(\frac{1}{2}, \frac{1}{2}\right)$, fifty-fifty randomization for both players.}
\label{fig:match-pennies-bimatrix}
\end{table}
Consider a finite memory learning process where the length of the history is $m = 9$, and both players sample the whole history and play without a mistake, i.e., $k=m$ and $\varepsilon = 0$. Assume that the history contains, reading from the oldest to the latest entry, four interactions where both players took action \textbf{1}, followed by five interactions where both took action \textbf{2}. The row player will then take action \textbf{2} and the column player action \textbf{1}. However, since the interaction that falls out of the history is one where the column player played \textbf{1}, 
the sample to which the row player responds will not change until the 
\textbf{1}:s in the end of the history have all fallen out and the first interaction with a \textbf{2} falls out of the history. At that point, the history contains five interactions where the column player played \textbf{1}, so now the row player wants to play \textbf{1} as well. However, by now all the interactions in the history are such that the row player played \textbf{2}. So for the coming five interactions they will both take action \textbf{1}. 
\begin{equation} 
	 \begin{pmatrix}
		111122222 
		\\
		111122222	
	\end{pmatrix}  
	\hspace{-.1cm}
	\to  
	\hspace{-.1cm}
	\begin{pmatrix}
	    222222222 
	    \\
	    222211111	
	\end{pmatrix} 
	\hspace{-.1cm}
	\to 
	\hspace{-.1cm}
	\begin{pmatrix}
    	222211111 
    	\\
    	111111111	
	\end{pmatrix}
	\hspace{-.1cm}
	\to
	\hspace{-.1cm}
	\begin{pmatrix}
    	111111111 
    	\\
    	111122222	
	\end{pmatrix}
	\hspace{-.1cm}
	\to
	\cdots
\end{equation}
The behavior in the next period depends as much on what falls out of the history as what is added, generating a cycling behavior. The cycling behavior does not only happen in this special case but is a general feature observed when simulating finite memory based learning processes, see Figure~\ref{fig:matchingpennies1} for a simulated example. In Figure~\ref{fig:recency-matchingpennies} we see that this cycling behavior is not appearing in the RWS dynamic. The lack of cycling is explained by the fact that past interactions never leave the memory of the RWS (which is defined without a memory limit), which when they do generates cycling in for example finite memory based learning models. 
\begin{figure}[H]
    \centering
    \includegraphics[width=0.8\textwidth]{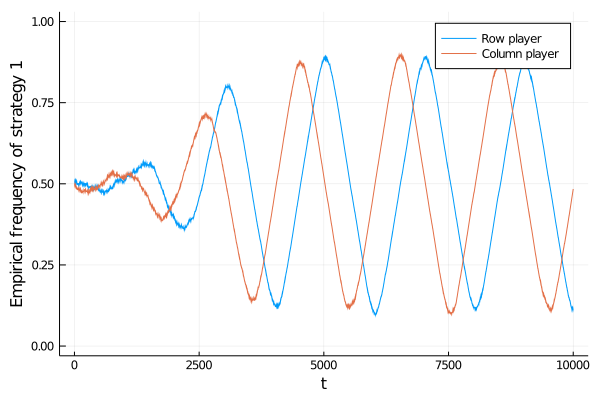}
    \caption{A 10 000 period simulation of Young's finite memory learning process on Matching Pennies with $m = 1000$, $k = 20$, $\varepsilon = 0.05$. Initiated in the mixed Nash equilibrium.}
    \label{fig:matchingpennies1}
\end{figure}
\begin{figure}[H]
	\centering
	\includegraphics[width=0.8\textwidth]{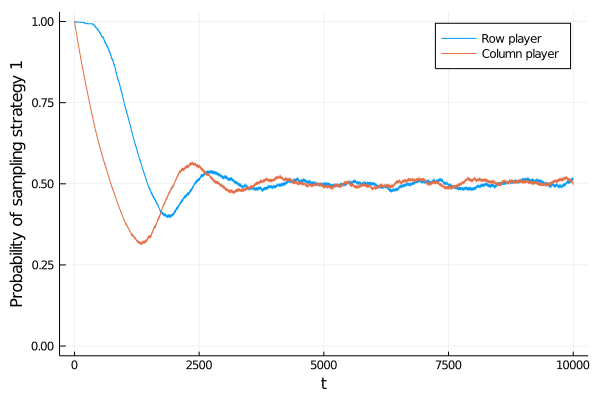}
	\caption{A 10 000 period simulation of the Recency Weighted Sampler on Matching Pennies with $\beta = 0.999$, $k = 20$, $\varepsilon = 0.05$. Initiated in the corner (1,1).}
	\label{fig:recency-matchingpennies}
\end{figure}

\section{Main Results}
\label{sec:mainresults}

\subsection{Ergodicity}
\label{sec:convergence}

Our first result is Theorem~\ref{thm:convergence_for_us} which states conditions for when the RWS state process is uniformly ergodic. Our proof relies on both the theory of Markov processes and functional analysis, and can be found in Appendix~\ref{sec:dinetan}. When $P$ is viewed as an operator on real-valued functions on $\square(S)$, it has a certain squeezing property that "flattens" Lipschitz continuous functions. Combined with the open-set accessibility this yields geometric convergence of RWS's distribution to a unique probability measure on $\square(S)$, invariant with respect to $P$. We state the theorem in the language of the theory of Markov processes where our sought property is called uniform ergodicity.

\begin{theorem}
\label{thm:convergence_for_us}
If $\varepsilon>0$ and $\beta \in (1-\max\{m_1,m_2\}^{-1},1)$, then the Markov chain with kernel $P$ is uniformly ergodic.
\end{theorem}

In other words, for whichever initial distribution $\nu\in\mathcal{P}(\square(S))$ the initial history $p(0)$ is sampled from, the distribution of $p(t)$ will converge "geometrically uniformly" as $t\rightarrow \infty$ to the probability measure $\mu^*_\varepsilon$ which is the unique solution of $\mu^*_\varepsilon P=\mu^*_\varepsilon $. More precisely, for every $\varepsilon\in(0,1]$ there exists a unique $\mu^\ast_\varepsilon\in\mathcal{P}(\square(S))$ such that $\mu^*P = \mu^*$ and for all $\alpha\geq 1$,
\begin{equation}
    \left(W_\alpha(\nu P^n, \mu^*_\varepsilon)\right)^\alpha 
    \leq c\theta^n, \quad  \nu\in\mathcal{P}(\square(S)),
\end{equation}
where $W_\alpha$ is the Wasserstein distance of order $\alpha$ between measures on $\square(S)$ (see Definition 6.1 in \cite{villani2008optimal}), $\theta\in (0,1)$, and $c$ is a positive constant depending only on $\max_{x\in\square(S)} |x|$ and $\alpha$.

The theorem is more general than what is needed for the goal of this paper. The result holds for any Markov chain with a compact state space and with a dynamic of the form \eqref{eq:RWS-dynamics}, as long as there is a positive lower bound for the probability that any strategy is played (in any state) and that this probability is Lipschitz continuous as a function of the state. Examples of best response functions to which Theorem~\ref{thm:convergence_for_us} applies are the logit best reply where
\begin{equation}
    \mathbb{P}\left[\widetilde{BR}_i (p) = s_i \right] 
    = \frac{\exp(\eta \pi_i(s_i,p_{-i}))}{\sum_{a \in S_i}\exp(\eta \pi_i(a,p_{-i}))},\quad \eta > 0,
\end{equation}
models where $k$ itself is a random parameter, and models where only robust best responses to the sample are considered. 

\subsection{Convergence to Minimal CURB Configurations}

Before turning to the convergence to  minimal CURB blocks, one minor technical detail most be resolved. A minimal CURB block is a
collection of strategy profiles $C = C_1 \times C_2 \subset S$  such that the best reply to all mixed strategies in the sub-simplex spanned by those strategies is always inside the spanning set, i.e. $BR(\sigma) \subset C$ for all $\sigma \in \square\left( C\right)$, where 
$\square\left( C\right):= \Delta\left(C_1\right)\times \Delta\left(C_2\right)$. However, since our agents only reply to samples of size $k$, it might be the case that the mixed strategy from the simplex that has a best reply outside a non-CURB block simply never is sampled. The game below is a simple illustration of this point. 
\begin{table}[H]
    \centering
    \begin{tabular}{c|rc|rc|}
    \multicolumn{0}{c}{}&  \multicolumn{2}{c}{$\mathbf{1}$} & \multicolumn{2}{c}{$\mathbf{2}$} 
    \\
    \cline{2-5}
    $\mathbf{1}$ & $2,$ &$-100$ &  $-100,$ &$2$ 
    \\ 
    \cline{2-5} 
    $\mathbf{2}$ & $-100,$ &$2$ &  $2,$ &$-100$ 
    \\ 
    \cline{2-5} 
    $\mathbf{3}$ & $1,$ &$0$ &  $1,$ &$0$ 
    \\ 
    \cline{2-5} 
    \end{tabular}
\end{table}

If $k=1$, then only the best replies to pure strategies will ever be considered. If the process initially has support only on the block 
$\{\mathbf{1},\mathbf{2}\} \times \{\mathbf{1},\mathbf{2}\}$, the best reply to any sample will be inside that block, even though $\mathbf{3}$ is the best reply to most properly mixed strategies. We could call this smaller set of blocks that are closed under best replies to any 
strategies on the $k$-lattice k-CURB blocks. In most settings, a relatively small $k$ is enough for the k-CURB blocks to coincide with the CURB blocks. In the rest of the paper, we will speak of CURB blocks and by that mean $k$-CURB blocks. Alternatively, one can think of $k$ as sufficiently large so that the notions coincide. 

In what follows, we first prove that the RWS concentrates (in probability) on minimal CURB blocks for general two player games. Then we prove the concentration of RWS paths to an approximate mixed Nash equilibrium for games with $m_1 = m_2 = 2$ and a unique mixed Nash equilibrium. 

\subsubsection{Concentration on Minimal CURB Blocks}

To prove concentration of the RWS on minimal CURB blocks we will partially rely on results for the original finite memory learning process. The RWS dynamics introduces some difficulties that are not present in the original model, mainly that once a strategy has been played it never truly disappears from memory but always has a positive probability of being sampled. However, the probability of sampling that strategy decreases over time as long as the strategy is not played again. A notion well-suited for the RWS is therefore the neighbourhood $B_{\delta}(C)$, $\delta>0$, of $C := C_1 \times C_2 \subset S$, defined as all pairs $(p_1,p_2)$ in $\square\left( S\right)$ such that each of the components puts at least $1 - \delta$ probability on the block $C$.

\begin{definition}
\label{def:area_around_C}
For all $\delta>0$,
\begin{equation}
    B_\delta(C_1\times C_2) 
    := \left\{ p=(p_1,p_2) \in \square\left( S\right) \ |\ \sum_{s=1}^{m_i}p_{i,s}1_{C_i}(s) \geq 1- \delta,\ i=1,2 \right\}.
\end{equation}
\end{definition}

Let $\mathscr{C}$ denote the union of all minimal CURB blocks in the game. To prove the concentration result Theorem~\ref{thm:min_curb}, we show that expected time to go from $B_\delta(\mathscr{C})^C$ to $B_\delta(\mathscr{C})$ is always bounded, but the expected time spent inside $B_\delta(\mathscr{C})$ once entered goes to infinity as $\varepsilon$ goes to zero. This in turn will imply that as $\varepsilon$ goes to zero, the invariant distribution concentrates on a neighbourhood of $\mathscr{C}$, the union of all minimal CURB blocks.

\begin{theorem}
\label{thm:min_curb}
If $\beta \in (1-\max\{m_1,m_2\}^{-1},1)$, then $\mu^*_\varepsilon$ concentrates on $B_\delta(\mathscr{C})$, $\delta > 0$, as $\varepsilon $ goes to zero,
\begin{equation}
    \lim_{\varepsilon\rightarrow 0}\mu^\ast_\varepsilon\left(B_\delta(\mathscr{C})\right) 
    = 1,\quad \delta > 0.
\end{equation}
\end{theorem}

\subsubsection{Behavior Inside Minimal CURB}
\label{sec:insideCURB}

We saw in the previous section that the RWS spends almost all the time inside minimal CURB blocks in the small-error limit, possibly with rare excursions between different minimal CURB blocks. As discussed in the introduction, it is known that many learning processes hold this property, but we had to reprove it for the RWS because of the continuous state space formulation. Theorem~\ref{thm:concetration_on_Nash} below further refines the long-run behaviour of the RWS for some games, it shows that the RWS concentrates mass on a neighbourhood of mixed Nash equilibria inside minimal CURB blocks. This refinement of conventions, which is clearly visible in simulations like those in the introduction, is not typically found in other models for learning and the main motivation for introducing the RWS.

Consider the deterministic mean-value process $x$, 
\begin{equation}
\label{eq:mean_val_proc}
    \dot{x}_i(t) 
    = \mathbb{E}\left[\widetilde{BR}_i(x_{-i}(t))\right]-x_i(t), \quad x_i(0) = p_i(0).
\end{equation}
The process in \eqref{eq:mean_val_proc} is a deterministic process that can be thought of as a continuous-time evolution of the expected value of the RWS state process \eqref{eq:RWS-dynamics}. If $p(0)$ is in a minimal CURB block the process \eqref{eq:mean_val_proc} converges to either a stable point or a stable orbit with constant distance to a stable point for $\varepsilon$ small enough, as a consequence of Lemma \ref{lemma:gas}. Furthermore, over a given time horizon divided the probability that the RWS trajectory is arbitrarily close to the deterministic process \eqref{eq:mean_val_proc} goes to $1$ as $\beta$ goes to 1, as shown in Lemma \ref{lemma:bdd-time-intervals} (a result we borrow from stochastic approximation theory). Taken together, if the deterministic process behaves well in the minimal CURB blocks of a game, the RWS's concentration around stable points or stable orbits is controlled by $\beta$. 

The next theorem states that for a $2 \times 2$  minimal CURB block with a unique mixed Nash equilibrium the RWS concentrates around the stable point of \eqref{eq:mean_val_proc}, which turns out to be unique. Its proof, and the proofs of the Lemmas cited above, are found in the appendix.

\begin{theorem}
\label{thm:concetration_on_Nash}
Let $G$ be a $2\times 2$ normal form game with a unique completely mixed Nash equilibrium. If $\beta > 1/2$, then there exists a positive constant $K$ such that
\begin{equation}
    \mu^*_\varepsilon\left(x\in\square(S) : \|x - x^*\|_\infty \geq \eta \right) 
    = O\left(\exp\left(-\frac{K\eta^2}{1-\beta}\right)\right),\quad \varepsilon > 0,\ \eta > 0,
\end{equation}
where $x^*$ is the unique stationary point of the mean-value process \eqref{eq:mean_val_proc}.
\end{theorem}

The stationary point of \eqref{eq:mean_val_proc} naturally depends on $k$. Under the assumptions in the theorem above, as $k \rightarrow \infty$ the equation $(\dot{x}_1(t), \dot{x}_2(t)) = (0,0)$ is satisfied only by the Nash equilibrium, call it $\hat x$, and we have that $\lim_{k\rightarrow \infty}x^* = \hat x$. So $x^*$ can be interpreted as a approximation of the Nash equilibrium.

The result of Theorem \ref{thm:concetration_on_Nash} can be extended to games of any size in the small-error limit $\varepsilon\rightarrow 0$, as long as the minimal CURB blocks that are either $1\times 1$ or $2 \times 2$ and satisfy the assumptions of Theorem \ref{thm:concetration_on_Nash}. We present a detailed argument in Appendix~\ref{sec:extension}.
 In generality, this can be extend to those types of minimal CURB blocks for which convergence of the deterministic process at a sufficient rate can be established.

\section{Conclusions and Outlook}
\label{sec:concl}

We have introduced a new learning process, the RWS, whose main feature is a procedure of sampling from historical interactions biased towards more recent events. We have shown that the RWS has several interesting properties. The invariant distribution of the RWS (which is a Markov process with continuous state space) concentrates on minimal CURB blocks in the small-error limit. So in the long run, the RWS will almost always be inside a minimal CURB block, perhaps with rare transitions between them. While the process is inside a minimal CURB block, the (deterministic) mean evolution of RWS will converge to either a stable point or a stable orbit, and the (stochastic) RWS state process will (with a high probability) not deviate far from it during any finite time horizon if the recency parameter $\beta$ is sufficiently close to 1. Combining these results we see that in the small-bias small-error double limit (as $\varepsilon$ and $\beta$ approach $0$ and $1$, respectively) the RWS almost always is in the neighbourhood of a stable point or a stable orbit inside a minimal CURB. Furthermore, since the sampling best reply function $\widetilde{BR}$ is continuous, we have that if the RWS state is close to some stable point, then so is play. 

For $2 \times 2$ minimal CURB blocks with a unique Nash equilibrium, we have shown that the deterministic mean process has a unique stable point which is close to the Nash equilibrium for most values of $k$ (in essence, values of $k$ that does not permit the players to be indifferent between strategies after sampling, cf. Lemma~\ref{eq:fp}). For games with minimal CURB blocks larger than $2 \times 2$, the picture is more complicated, and it is beyond the scope of this paper to completely map it out. However, for small to intermediate $k$ the RWS behaves well, at least numerically, when other learning dynamics does not. Consider the unstable rock paper scissors game, see Table~\ref{tab:RPS}, studied in e.g. \cite{Benaim2009}.

\begin{table}[H]
\centering
    \begin{tabular}{c|c|c| c|}
    \multicolumn{0}{c}{}& \multicolumn{0}{c}{R} & \multicolumn{0}{c}{P} & \multicolumn{0}{c}{S} 
    \\
    \cline{2-4}
    R & $0, 0$ &  $-3, 1$ & $1, -3$ 
    \\ 
    \cline{2-4} 
    P & $1, -3 $ & $0, 0$ & $-2, 1$ 
    \\ 
    \cline{2-4}
    S & $-3, 1$ & $1, -2$ & $0,0$ 
    \\ 
    \cline{2-4}
    \end{tabular}
\caption{The payoff in the Unstable Rock Paper Scissors game. The unique symmetric Nash equilibrium is  $\left( \frac{9}{32}, \frac{10}{32}, \frac{13}{32} \right)$.}
\label{tab:RPS}
\end{table}
Classical learning processes such as fictitious play or reinforcement learning circles the Nash equilibrium in a stable cycle. In Figure \ref{fig:unstable_RPS} we compare the performance of RWS with $k=20$ and fictitious play with recency. The RWS remains close to the equilibrium over time, even in this unstable game, while the fictitious play dynamic circles the equilibrium. When $k$ is larger the RWS behaves as fictitious play with recency. This is expected, as $k$ grows the sampled beliefs $(D_1, D_2)$, see \eqref{eq:samples}, become more and more similar to the sampling probabilities by the law of large numbers.
 \begin{figure}[H]
 \centering
     \begin{subfigure}[b]{0.49\textwidth}
     \centering
     \includegraphics[width=\textwidth]{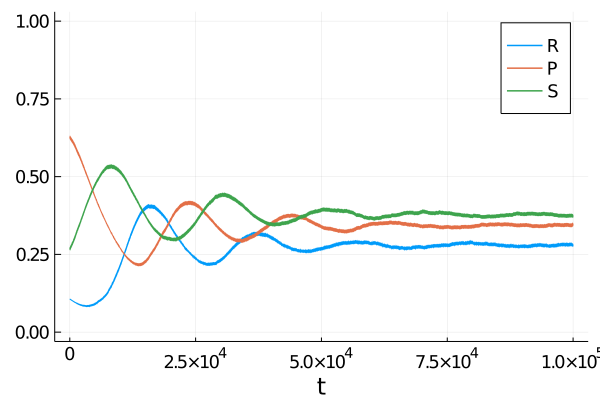}
     \caption{RWS with $\varepsilon = 0$ and $k=20$.}
     \end{subfigure}
     \hfill
     \begin{subfigure}[b]{0.49\textwidth}
     \centering
     \includegraphics[width=\textwidth]{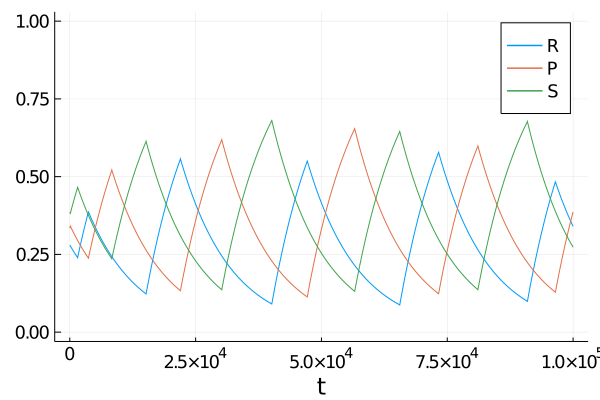}
     \caption{Fictitious play with recency.}
     \end{subfigure}
\caption{Simulations of behavior in the Unstable Rock Paper Scissors game. \textit{Left:} RWS with a low $k$-value and no noise. \textit{Right:} fictitious play with recency. The recency parameter was set to $\beta=0.9999$ in both simulations.}
\label{fig:unstable_RPS}
\end{figure}
We can numerically validate the bound from Theorem~\ref{thm:concetration_on_Nash} for the Unstable Rock Paper Scissors game. We compute the probability of $\|p(t)-x^*_{RPS}\|_\infty \leq \eta$ under $\mu^*_\varepsilon$, where $x^*_{RPS}$ is the Nash equilibrium (see Table~\ref{tab:RPS}), for values of $\beta$ in $[0.9,1]$ and for $\eta = 0.01$. The decay the probability with increasing $\beta$ is compared to the theoretical bound $o(\exp(-K\eta^2/(1-\beta))$ in Figure~\ref{fig:eta_sim}, and we see a clear agreement of the simulated probability and the theoretical bound.
\begin{figure}[H]
\centering
    \includegraphics[width=0.7\linewidth]{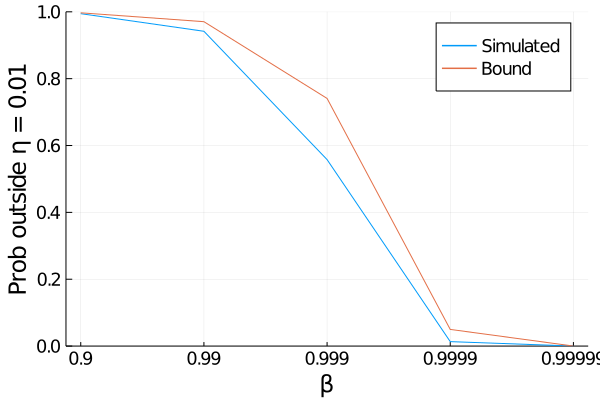}
    \caption{The estimated probability of the process $p$ being further away than $\eta=0.01$ in a given time period, and the bound with $K=3$. }
    \label{fig:eta_sim}
\end{figure}

This paper leaves some obvious questions about the RWS learning process open, which we hope to address in future work. For example, which minimal CURB configurations will have positive measure in the long run? We expect that this question can be approached using standard radius and co-radius arguments as in \cite{Ellison2000} or \cite{benaim2003deterministic}. On a final note, we expect that the results for games with minimal CURB blocks of size $1\times 1$ and $2\times 2$ generalize to games with minimal CURB blocks that are interior ESS, zero-sum games, potential games, and supermodular games, e.g., by an argument similar to that in \cite{Hofbauer2002}. There, the converge of a mean process (with close resemblence to the RWS mean process) to Nash equilibrium is proven using Lyapunov theory. Convergence rates are not derived explicitly. The rates are needed in the analysis of this paper to balance the time the process spends in a neighbourhood of Nash equilibrium with the escape probability from a minimal CURB block. It is however often possible to find convergence rates in Lyapunov stability analysis. Exploring this direction will be the topic of future work.

\appendix

\section{The Basic Properties of the Learning Process}
\label{sec:prof1}

\subsection{Exponential History}
\label{app:exponential_history}

Let us prove Proposition~\ref{exphist}. Starting from the definition, we have
\begin{equation}
    p_{-i,s}\left( t+1\right) 
    = \left( 1-\beta\right) \sum_{\tau=1}^{\infty}\beta^{\tau-1} 1_{s}(s_{-i} (t-\tau+1)).
\end{equation}
After index substitution $v=\tau-1$, splitting the term $v=0$ yields
\begin{equation}
    p_{-i,s}\left( t+1\right) 
    = \left( 1-\beta\right) \left( 1_s(s_{-i}(t)) + \sum_{v=1}^{\infty}\beta^{v} 1_s(s_{-i}(t-v)) \right).
\end{equation}
In other words,
\begin{equation}
    p_{-i,s}\left( t+1\right) 
    = \left( 1-\beta\right) \beta \sum_{v=1}^{\infty }\beta^{v-1}
    1_s(s_{-i}(t-v)) + \left( 1-\beta\right) 1_s(s_{-i}(t)).
\end{equation}
We recognize the first term as $p_{-i,s}\left( t\right) $, so we are left for every $s \in S_{-i}$ with
\begin{equation}
    p_{-i,s}\left( t+1\right) 
    = \beta p_{-i,s}\left( t\right) +\left(1-\beta\right) 1_s(s_{-i}(t)),
\end{equation}
which is the representation we seek.

\subsection{Lipschitz Continuity}

\begin{lemma}
\label{lemma:lipschitz}
For all $k\in \mathbb{N}$, $i\in\{1,2\}$, and $a \in \{1,\dots, m_i\}$,
\begin{equation}
    \Delta(S_{-i}) \ni p \to \mathbb{P}\left(\widetilde{BR}_i(p) = a\right)
\end{equation}
is Lipschitz continuous with Lipschitz coefficient at most $(1-\varepsilon)km_{-i}$.
\end{lemma}

\begin{proof}
At the beginning there is a sample with respect to probabilities $p$, yielding a random vector $N:=\left(n_{-i,1}\left( t\right) ,\ldots n_{-i,m_{-i}}\left( t\right) \right) $ of integers from the (discrete) probability distribution%
\begin{equation}
    \mathbb{P}\left( N=\left( n_{1},\ldots n_{m_{-i}}\right) \right)
    = k!\prod\limits_{j=1}^{m_{-i}}\frac{p_{s}^{n_{s}}}{n_{s}!}.
\end{equation}
Each $N$ will lead to an empirical opposing strategy profile $D$, that must belong to some finite 'simplex grid'
\begin{equation}
    \Delta ^{\left(-i, k\right) } 
    := \left\{ \frac{1}{k}\sum_{s \in S_{-i}}n_{s} \overrightarrow{1_{-i,s}}\ ;\ n_{s}\in \mathbb{N}_0,\, \sum_{s \in S_{-i}}n_{s}=k\right\}.
\end{equation}
Now let us form $m_{i}$ subsets from $\Delta ^{\left(-i, k\right) }$ (which is finite), named $\Delta _{s}^{\left(-i,  k\right) }$ for $s \in S_i$, where $x\in \Delta _{s}^{\left(-i, k\right) }$ whenever $s\in BR_i(x)$. Note that $(\Delta^{(-i,k)}_s)_s$ is not a disjoint cover of $\Delta^{(-i,k)}$ except in the special case when each $x \in \Delta^{(-i,k)}$ has a unique best response. Also, $\cup_s \Delta^{-i,k}_s = \Delta^{-i,k}$ since the best response set is never empty.

For $a\leq m_{i}$, the probability that $\widetilde{BR}_{i} \left( p\right) = a$ is going to be played is thus obtained as follows :
\begin{itemize}
    \item If the player $i$ trembles, which happens a fraction $\varepsilon $ of the time, strategy $a$ is played with a probability $1/m_i$, totalling $\varepsilon /m_i$.
    \item Otherwise the player selects its best response, so it will be $a$ with the probability $\mathbb{P}\left( D\in \Delta _{a}^{\left(-i, k\right) },\, \widehat{BR}_{i}(D) = a \right) $.
\end{itemize}
In short,
\begin{equation}
\label{eq:in_short_eq}
    \mathbb{P}\left( \widetilde{BR}_{i}\left( p\right) = a\right) 
    = \varepsilon r_{a}+\left( 1-\varepsilon \right) \sum_{x\in \Delta _{a}^{\left(-i, k\right) }}\mathbb{P}\left( \widehat{BR}_{i}(x) = a\right)\mathbb{P}\left( D=x\right).
\end{equation}
However $D=x$ is an event of the shape $N=\left( n_{1},\ldots n_{m_{-i}}\right) $, so considering $\mathbb{P}\left( D=x\right) $ as a
function of $p_{1},\ldots p_{n}$, we get 
\begin{equation}
    \frac{\partial \mathbb{P}\left( N=\left( n_{1},\ldots n_{m_{-i}}\right) \right) }{\mathbb{\partial }p_{b}}
    = k!\frac{p_{b}^{n_{b}-1}}{\left( n_{b}-1\right) !}\prod\limits_{j\neq b}\frac{p_{j}^{n_{j}}}{n_{j}!},
\end{equation}
with the convention $1/\left( -1\right) !=0$ for continuity. So relatively to the norm $\left\Vert .\right\Vert _{\infty }$ over $\Delta \left( S_{-i}\right) $, the Lipschitz constant of the probabilities $\mathbb{P} \left( D\in \Delta _{a}^{\left(-i, k\right) }\right) $ are at most
\begin{equation}
    \sum_{b=1}^{m_{-i}}\left\vert \frac{\partial \mathbb{P}\left( D\in \Delta_{a}^{\left(-i, k\right) }\right) }{\mathbb{\partial }p_{b}}\right\vert 
    \leq \sum_{b=1}^{m_{-i}}\sum_{x\in \Delta _{a}^{\left( -i,k\right) }}k!\frac{p_{b}^{n_{b}-1}}{\left( n_{b}-1\right) !}\prod\limits_{\substack{j=1 \\ j\neq b}}^{m_{-i}}\frac{ p_{j}^{n_{j}}}{n_{j}!}.
\end{equation}
However we know that
\begin{equation}
    \sum_{x\in \Delta^{\left(-i,k\right)}}\frac{p_{b}^{n_{b}-1}}{\left( n_{b}-1\right) !}\prod\limits_{\substack{j=1 \\ j\neq b}}^{m_{-i}}\frac{p_{j}^{n_{j}}}{n_{j}!}
    = \frac{1}{\left( k-1\right) !},
\end{equation}
as this is the multinomial formula for $k-1$ draws. Since $\Delta
_{a}^{\left(-i, k\right) }\subset \Delta ^{\left( -i,k\right) }$, the Lipschitz constant of $\mathbb{P}\left( D\in \Delta _{a}^{\left(-i, k\right) }\right) $ is at most
\begin{equation}
    \sum_{b=1}^{m_{-i}}k!\frac{1}{\left( k-1\right) !}
    = km_{-i}.
\end{equation}
Bounding $\mathbb{P}\left( \widehat{BR}_i(x) = a \right)$ from~\eqref{eq:in_short_eq} by 1, the Lipschitz constant for \begin{equation}
    p \mapsto \mathbb{P}\left( \widetilde{BR}_{i}\left( p\right) = a\right)
\end{equation}
is at most $\left( 1-\varepsilon \right) km_{-i}$.
\end{proof}

\subsection{Proof of Theorem~\ref{thm:convergence_for_us}}
\label{sec:dinetan}

The proof relies on arguments from functional analysis, therefore we introduce here some function spaces. Let $\mathcal{C} := C(\square(S); \mathbb{R})$, the set of real-valued continuous functions on $\square(S)$, endowed with $\| \cdot \|_\infty$ the uniform norm. Let $\mathcal{L}\subset \mathcal{C}$ be the subspace of Lipschitz continuous functions over $\square(S)$, endowed with the seminorm $\| \cdot \|_{\mathcal{L}}$ 
\begin{equation}
    \| f\|_\mathcal{L} 
    := \sup_{x,x'\in \square(S): x\neq x'}\left(|f(x) - f(x')|/|x-x'|\right).
\end{equation}
The norm $\| f \|_{\mathcal{Z}} := \max (\|f\|_\infty, \|f\|_{\mathcal{L}})$ makes $\mathcal{L}$ a Banach space. The unit ball in $(\mathcal{L},\|\cdot\|_{\mathcal{Z}})$ is compact. 

The RWS transition kernel $P$ can be interpreted as an operator from $\mathcal{C}$ to $\mathcal{C}$
\begin{equation}
\label{eq:def_of_operator}
    \begin{aligned}
    P: \mathcal{C} &\rightarrow \mathcal{C}
    \\
    f\left( x \right) 
    &\mapsto \sum_{s_{1}=1}^{m_{1}}\sum_{s_{2}=1}^{m_{2}}
    \mathbb{P}\left( \sigma(x,(s_1,s_2))\right) f\left(\Gamma(x,(s_1,s_2)) \right).
    \end{aligned}
\end{equation}
where we define for $x = (x_1,x_2) \in \square(S)$ and $s = (s_1,s_2) \in S$
\begin{equation}
    \begin{aligned}
    \sigma(x,s) 
    &:= \left( \widetilde{BR}_{1}\left( x_{2}\right) =s_{1},\widetilde{BR}_{2}\left( x_{1}\right) =s_{2}\right)
    \\
    \Gamma(x,s) 
    &:= \left(\beta x_{1}+\left( 1-\beta\right) \overrightarrow{
    1_{1,s_{1}}},\beta x_{2}+\left( 1-\beta\right) \overrightarrow{
    1_{2,s_{2}}} \right).
    \end{aligned}
\end{equation}
As an operator, $P$ is linear and continuous. However, it is not a compact operator (as $P$ maps continuous functions to a linear combination of continuous functions, the image of the unit ball in $\mathcal{C}$ under $P$ is not a compact subset of $\mathcal{C}$). The lack of compactness prevents us from using for example the Krein-Rutman theorem to determine ergodicity. A more suitable set to use for the study of convergence is instead the set of Lipschitz continuous functions over $\square(S)$. 

The core idea of this appendix is to show that for a Lipschitz-continuous function $f$, our operator squeezes $f$ in the sense that $\| P^nf\|$ tends to a constant function as $n$ tends to infinity. We take the following path to a proof of Theorem~\ref{thm:convergence_for_us}: we get open set accessibility of $P$ (viewed as a Markov kernel) in \ref{subsubsec:lower_bound} by first deriving a useful lower bound in \ref{ref:infra}. The squeezing property of $P$ is derived in \ref{sec:lipschitz-friendly} and used in \ref{subsubsec:geometric} to show that repeated application of $P$ eventually results in a constant function, at a geometric rate. Finally, the previous steps are combined to conclude uniform ergodicity in \ref{subsubsec:ergodicity}.\footnote{Some of the arguments are inspired by an unpublished manuscript written by Lee Dinetan, which studies Krein-Rutman like theorems for non-compact operators with the squeezing property described above.}

\subsubsection{Approximation of History with Truncated Paths}
\label{ref:infra}

For $i\in \left\{ 1,2\right\} $, $j\in\{1,\dots,m_{i}\}$, and $t\in \mathbb{N}$, let $\omega _{i,j,t} : =1_j\left( s_{i}\left( t\right)\right) $ be the indicator of a play $j$ by player $i$ at time $t$, so that 
\begin{equation}
    p_{i,j}\left( t\right) 
    = \left( 1-\beta\right) \sum_{\tau=1}^{\infty }\beta^{\tau-1}\omega _{i,j,t-\tau}.
\end{equation}
We will call $\Sigma ^{\left( i\right) }:=\left\{ 0,1\right\}
^{m_{i}\times \mathbb{N}}$ the set of binary arrays, indexed by $s\in\{1,\dots, m_{i}\}$ and $t\in \mathbb{N}$, such that for every $t$ there is exactly one $s$ such that $\Sigma _{s,t}^{\left( i\right) }=1$. In other words, $\Sigma ^{\left( i\right) }$ represents a possible history for player $i$, where $1$ at the entry $\left( s,t\right) $ indicates that $s$ was played at time $t$. Likewise, for $n\in \mathbb{N}$, we will call $\Sigma ^{\left(i,N\right) }:=\left\{ 0,1\right\} ^{m_{i}\times N}$ the set of binary arrays indexed by $s\in\{1,\dots, m_{i}\}$ and $t\in \left\{ 1,\ldots N\right\} $
obeying the same condition, in other words the history up to time $N$.

Let $p_{i}\in \Delta \left( S_{i}\right) $. We are going to exhibit a sequence of plays of finite length $N$, i.e., an $\omega \in \Sigma ^{\left(i,N\right) }$ for some $N\in \mathbb{N}$, such that the partial sum
\begin{equation*}
    p_{i,j}^{\left( N\right) }
    := \left( 1-\beta \right) \sum_{\tau=1}^{N}\beta^{t-1}\omega_{i,j,\tau}
\end{equation*}
falls close to $p_{i}$. Namely, we want to prove the following.

\begin{lemma}
\label{apphis}
Let $p_{i}\in \Delta \left( S_{i}\right) $ and $\delta >0$. We assume
that $\left( 1-\beta\right) m_{i}\leq 1$. There exists an $N(\delta)\in
\mathbb{N}$, independent of $i$ and $p_i$, such that there is a history $\omega_i^{(N)} \in \Sigma^{(i,N)}$ for each $N\geq N(\delta)$ which satisfies
\begin{equation}
\label{eq:lemma_13_eq}
    p_{i,j}^{\left( N\right) } 
    = \left( 1-\beta\right)\sum_{\tau=1}^{N}\beta^{\tau-1}\omega^{(N)}_{i,j,\tau} \in \left(\max\{p_{i,j}-\delta,0\} ,p_{i,j}\right]
\end{equation}
for all $j\in\{1,\dots,m_i\}$.
\end{lemma}

\begin{proof}
The following algorithm provides a proof of Lemma~\ref{apphis}. 
Start by setting $p_{i,j}^{\left( 0\right) } =  0$ for all $j=1,\dots, m_i$, and $\omega^{(0)}_{i}$ to the empty array of dimensions $0$ and $m_{i}$. Define $N(\delta)$ as the smallest $N\in \mathbb{N}$ such that $\beta^{N}< \delta $, i.e., 
\begin{equation}
    N(\delta) := \inf\{N\in\mathbb{N} : \beta^N < \delta \}.
\end{equation}
For $t\in \left\{ 1,\ldots N(\delta)\right\} $, repeat the following steps:
\begin{enumerate}
    \item Look for the indices $j \in \{1,\dots, m_i\}$ such that $p_{i,j}-p_{i,j}^{\left(t-1\right) }$ is maximal, and call any of these indices $a$.
    \item Append $\overrightarrow{1_{1,a}}$ to $\omega_{i}^{(t)}$. Now $\omega_{i,a,t}^{(t)}=1$ and $\omega^{(t)}_{i,j,t}=0$ for $j\neq a$.
    \item Compute $p_{i,j}^{\left( t\right) }$ accordingly to~\eqref{eq:lemma_13_eq} and the updated history $\omega^{(t)}_{i}$.
\end{enumerate}
Return the final history $\omega^{(N(\delta))}_{i}$ and values $p_{i,j}^{\left(N(\delta)\right) }$.

We are going to prove inductively that for every $t\in \mathbb{N}$, we
always have 
\begin{equation}
\label{eq:induction_ineq}
    p_{i,j}^{\left( t\right) }\leq p_{i,j},\quad j=1,\dots,m
\end{equation} 
and
\begin{equation}
\label{eq:induction_eq}
    \sum_{j=1}^{m_{i}}p_{i,j}^{\left( t\right) } = 1-\beta^{t}.
\end{equation}
For $t=0$, \eqref{eq:induction_ineq} is true since $p_{i,j}$ is non-negative and $p_{i,j}^{\left( 0\right) }=0$ for all $j=1,\dots,m_i$, which also yields that \eqref{eq:induction_eq} holds at $t=0$. Now assume that~\eqref{eq:induction_ineq}--\eqref{eq:induction_eq} hold at time $t$. Since $\sum_{j=1}^{m_i} p_{i,j} = 1$, the maximal difference $\max_{1\leq j\leq m_i}(p_{i,j}-p_{i,j}^{\left(t\right)})$ must be at least $\beta^{t}/m_{i}$. By definition, then $\omega_{i,a,t+1}=1$ for some $a\in\{2,\dots, m_i\}$ and
\begin{equation*}
    \begin{aligned}
    p_{i,a}^{\left( t+1\right) }
    &= (1-\beta)\sum_{\tau=1}^{t+1}\beta^{\tau-1}\omega^{(t+1)}_{i,a,\tau}
    \\
    &= (1-\beta)\beta^t\omega^{(t+1)}_{i,a,t+1} + (1-\beta)\sum_{\tau=1}^t \beta^{\tau-1}\omega^{(t)}_{i,a,\tau}
    \\
    &= \left( 1-\beta \right)\beta^t + p_{i,a}^{\left( t\right) }
    \\
    &\leq \left((1-\beta)-\frac{1}{m_i}\right)\beta^t + p_{i,a}
    \end{aligned}
\end{equation*}
Therefore, since $\left( 1-\beta\right) m_{i}\leq 1$ as assumed, the
right-hand side is also bounded by $p_{i,j}$. As for other strategies $j\neq a$, since $p_{i,j}^{\left( t+1\right) }=p_{i,j}^{\left( t\right) }$ the inequality $p_{i,j}^{\left( t+1\right) }\leq p_{i,j}$ holds and we have proven the induction step for \eqref{eq:induction_ineq}. Now we also know that
\begin{equation*}
    p_{i,j}^{\left( t+1\right) }-p_{i,j}^{\left( t\right) }
    = \left( 1-\beta \right) \beta^{t}\omega^{(t+1)}_{i,j,t+1},
\end{equation*}
and since exactly one among the $m_{i}$ entries in $\omega^{(t+1)}_{i,t+1}$ is $1$, the other being zero, we have
\begin{equation*}
    \sum_{j=1}^{m_{i}}\left( p_{i,j}^{\left( t+1\right) }-p_{i,j}^{\left( t\right) }\right) 
    = \left( 1-\beta\right) \beta^{t}.
\end{equation*}
The induction hypothesis thus leads us to%
\begin{equation*}
    \sum_{j=1}^{m_{i}}p_{i,j}^{\left( t+1\right) }
    = 1-\beta^{t}+\left(1-\beta\right) \beta^{t}=1-\beta^{t+1},
\end{equation*}
which proves \eqref{eq:induction_eq} by induction. So in particular after time $N(\delta)$, by choice of $N(\delta)$, for every $ N\geq N(\delta)$ we have
\begin{equation*}
    \sum_{j=1}^{m_{i}}p_{i,j}^{\left( N\right) } > 1-\delta,
\end{equation*}
while $p_{i,j}^{\left( N\right) }\leq p_{i,j}$ for every $j$. Since $\sum_j p_{i,j} = 1$, this is possible only if $p_{i,j}^{(N)}>p_{i,j}-\delta$ for each $j =1,\dots, m_i$, leading to the result.
\end{proof}

\subsubsection{Open Set Accessibility}
\label{subsubsec:lower_bound}

Let $x=\left( x_{1},x_{2}\right)\in \square (S)$. We apply Lemma~\ref{apphis} with $\delta = \varepsilon$ to $x_{1}$ and $x_{2}$,
yielding play records $\omega^{(N(\varepsilon))}_{1}$ and $\omega^{(N(\varepsilon))}_2$, and values $p_{i,j}^{\left( N(\varepsilon)\right) }$ such that for every $1\leq j \leq m_{i}$ and $N\geq N(\varepsilon)$,
\begin{equation}
    p_{i,j}^{\left( N\right) }\in \left( \max\{x_{i,j}-\varepsilon,0\} , x_{i,j}\right].
\end{equation}
We know that for any $B\in \mathcal{B}(\square(S))$,
\begin{equation}
    P\left( x, B\right)
    = \sum_{s_{1}=1}^{m_{1}}\sum_{s_{2}=1}^{m_{2}}\sigma \left( x,s\right) 1_B\left( \Gamma \left( x,s\right) \right),
\end{equation}
and $\sigma(x,s)$ is uniformly bounded from below by $\eta = \varepsilon/(m_1m_2)> 0$.

The play records $\omega^{(N)}_{1}$ and $\omega^{(N)}_{2}$ up to time $N$ from the previous Lemma are now read in reverse time order. At each time step $t\in \left\{ 0,\ldots N-1\right\} $, there is a probability at least $\eta^{2}$ that player $1$ chooses the strategy $1\leq a\leq m_{1}$ given by $\omega^{(N)} _{1,a,N-t}=1$, and player $2$ chooses the strategy $1\leq b\leq m_{2}$ given by $\omega^{(N)}_{2,b,N-t}=1$. Therefore the plays up to time $N$ have a probability at least $\eta ^{2N}>0$ of being dictated by $\omega^{(N)}_{1}$ and $\omega^{(N)}_{2}$. When this happens, thanks to the Proposition~\ref{exphist}, a history having started by $\left( p_{1}\left( 0\right) ,p_{2}\left( 0\right) \right) =p$ will now be at the position
\begin{eqnarray}
\label{eq:sum_in_Lip_proof}
    \left( p_{1}(N), p_{2}( N) \right) 
    &=& \sum_{t=1}^{N}\left( \left( 1-\beta \right) \beta^{t-1} \omega^{(N)}_{1,t},\left( 1-\beta \right) \beta ^{t-1} \omega^{(N)}_{2,t} \right)  
    \\
    && +\left( \beta^{N}p_{1}\left( 0\right) ,\beta^{N}p_{2}\left(0\right) \right),
\end{eqnarray}
with probability greater or equal to $\eta^{2N}$.

By Lemma~\ref{apphis}, the choice of the records $\omega^{(N)}_{1}$ and $\omega^{(N)}_{2}$ makes $j$:th component of the sum on the right-hand side of \eqref{eq:sum_in_Lip_proof} take some value between 
$(\max\{x_{1,j}-\varepsilon,0\},\max\{x_{2,j}-\varepsilon, 0\})$ and $(x_{1,j}, x_{2,j})$. As we also have $\beta^{N}<\varepsilon $ and $p_{i,j}\left( 0\right) \leq 1$, we get $p_{i,j}\left( N\right) \in \left( x_{i,j}- \varepsilon, x_{i,j} + \varepsilon\right)$. We conclude that for all $N\geq N(\varepsilon)$
\begin{equation}
\label{useful_lower_bound}
    \mathbb{P}\left(|p(N) - x | < \varepsilon\right) \geq \eta^{2N}.
\end{equation}
In other words, it means that the point $y=\left( p_{1}\left( N\right)
,p_{2}\left( N\right) \right)$, which is in an $\varepsilon$-neighbourhood of $x$, is accessible from $p$ in $N$ steps. 

\subsubsection{Lipschitz Friendliness Property}
\label{sec:lipschitz-friendly}

The next lemma proves that $P$ has a "squeezing" property on Lipschitz functions in the following sense: there are $\lambda^+ > \lambda^- \in \mathbb{R}_+$ such that for every $f\in \mathcal{C}$ holding $\|f\|_{\mathcal{L}}\leq \lambda^+$ and $\|f\|_\infty \leq 1$, then $\|Pf\|_{\mathcal{L}} \leq \lambda^-$ and $\|Pf\|_\infty \leq 1$. We say that $P$ is "Lipschitz-friendly" since it has this property. 

\begin{lemma}
There is a positive constant $c$ such that $P$ is \textit{Lipschitz-friendly} if $\lambda^+ = \frac{2c}{1-\beta}$ and $\lambda^- = \beta\lambda^+ + c$. 
\end{lemma}

\begin{proof}
Let $f$ hold these properties in the Lemma statement, let $p,p^{\prime }\in \square\left(S\right)$, and consider the difference
\begin{equation}
\label{eq:summand_LPF}
    Pf\left( p^{\prime }\right) - Pf\left( p\right) 
    = \sum_{s\in S}\sigma \left(p^{\prime },s\right) f\left( \Gamma \left( p^{\prime },s\right) \right)-\sigma \left( p,s\right) f\left( \Gamma \left( p,s\right) \right).
\end{equation}
Here, let us use the identity
\begin{equation}
    ab-cd
    = \frac{1}{2}\left( a+c\right) \left( b-d\right) +\frac{1}{2}\left(
    a-c\right) \left( b+d\right)
\end{equation}
to transform the summand in \eqref{eq:summand_LPF} into%
\begin{equation}
\label{eq:summand_LPF_transformed}
    \begin{aligned}
    &\frac{1}{2}\left( \sigma \left( p^{\prime },s\right) +\sigma \left(p,s\right) \right) \left( f\left( \Gamma \left( p^{\prime },s\right) \right)-f\left( \Gamma \left( p,s\right) \right) \right)
    \\
    &\qquad +\frac{1}{2}\left( \sigma \left( p^{\prime },s\right) -\sigma \left(p,s\right) \right) \left( f\left( \Gamma \left( p^{\prime },s\right) \right)+f\left( \Gamma \left( p,s\right) \right)\right).
    \end{aligned}
\end{equation}
On the top row of \eqref{eq:summand_LPF_transformed}, from $\Gamma $ Lipschitz of constant $\beta $, we have
\begin{equation}
    \left\vert f\left( \Gamma \left( p^{\prime },s\right) \right) -f\left(\Gamma \left( p,s\right) \right) \right\vert 
    \leq \left\Vert p^{\prime}-p\right\Vert _{\infty }\beta \Vert f\Vert _{L}.
\end{equation}
Therefore we get
\begin{eqnarray}
    &&\left\vert \sum_{s\in S}\left( \sigma \left( p^{\prime },s\right) +\sigma\left( p,s\right) \right) \left( f\left( \Gamma \left( p^{\prime },s\right)\right) -f\left( \Gamma \left( p,s\right) \right) \right) \right\vert 
    \\
    &\leq& \sum_{s\in S}\left( \sigma \left( p^{\prime },s\right) +\sigma \left(p,s\right) \right) \left\Vert p^{\prime }-p\right\Vert _{\infty }\beta \Vert f\Vert _{L} 
    \\
    &=& 2\beta \Vert f\Vert _{L}\left\Vert p^{\prime }-p\right\Vert_{\infty }.
\end{eqnarray}
The last equality holds since the sum is taken over a disjoint partition of the space of outcome.

As a function of only the first argument $\sigma$ is Lipschitz continuous, uniformly over $s$, with Lipschitz constant at most
\begin{equation*}
    \left( 1-\varepsilon \right) k \left( m_{1} + m_{2} \right).
\end{equation*}
To simplify notation, denote the Lipschitz constant by $L_\sigma$. For the bottom row of \eqref{eq:summand_LPF_transformed} we then have the estimate
\begin{equation*}
    \left\vert \sigma \left( p^{\prime },s\right) -\sigma \left( p,s\right) \right\vert 
    \leq L_\sigma\left\Vert p^{\prime }-p\right\Vert _{\infty },
\end{equation*}
and we get
\begin{eqnarray*}
    &&\left\vert \sum_{s\in S}\left( \sigma \left( p^{\prime },s\right) -\sigma\left( p,s\right) \right) \left( f\left( \Gamma \left( p^{\prime },s\right)\right) +f\left( \Gamma \left( p,s\right) \right) \right) \right\vert 
    \\
    &\leq& \sum_{s\in S} L_\sigma \left\Vert p^{\prime }-p\right\Vert _{\infty } 2\left\Vert f\right\Vert _{\infty } 
    \\
    &\leq& 2L_\sigma\left\Vert f\right\Vert _{\infty }m_{1}m_{2}\left\Vert p^{\prime }-p\right\Vert _{\infty }.
\end{eqnarray*}
Let us denote
\begin{equation*}
    \lambda := \max \left\{ 1,m_{1}m_{2} L_\sigma\right\},
\end{equation*}
so that in the end, we have
\begin{equation*}
    \left\vert Pf\left( p^{\prime }\right) - Pf\left( p\right) \right\vert 
    \leq \left( \beta \Vert f\Vert _{L}+\lambda\left\Vert f\right\Vert _{\infty }\right) \left\Vert p^{\prime }-p\right\Vert _{\infty }.
\end{equation*}
Therefore, whenever $\left\Vert f\right\Vert _{\infty }\leq 1$ and
\begin{equation*}
    \Vert f\Vert _{L}\leq \lambda ^{+}:=\frac{2\lambda}{1-\beta },
\end{equation*}
the Lipschitz constant of $Pf$ will be at most $\beta \lambda
^{+}+\lambda=:\lambda ^{-}<\lambda ^{+}$ by choice of $\lambda ^{+}$. It follows that $P$ is Lipschitz-friendly with these $\lambda ^{+}$ and $\lambda ^{-}$.
\end{proof}

\subsubsection{Geometric Convergence of the Spread of Lipschitz Functions}
\label{subsubsec:geometric}

Denote by $B$ the subset of $\mathcal{L}$ such that $f\in B$ if $\|f\|_\infty \leq 1$ and $\|f\|_{\mathcal{L}} \leq \lambda^+$. $B$ is compact in $\mathcal{L}$ (endowed with the norm $\|\cdot\|_{\mathcal{Z}}$) by the Arzelà–Ascoli theorem. Essentially, $B$ is the set that is squeezed by $P$. $P$ squeezes $B$ into a compact set $C_1 := \{Pf; f\in B\} \subset B$. Also, $C_{n+1} := \{Pf; f\in C_{n}\} \subset C_n$ are compact. 

We define the spread of a function $f\in \mathcal{C}$ to be
\begin{equation*}
    [f] := \max_{x\in\square(S)} f(x) - \min_{x\in \square(S)} f(x)
\end{equation*}

\begin{lemma}
\label{lemma:B}
There exists constants $c\in (0,1)$ and $n\in \mathbb{N}$ such that for all non-negative $f\in C_1$:
$$[P^nf]\leq c.$$
\end{lemma}

\begin{proof}
Let $C_1^+ := \{f \in C_1 : f(x) \geq 0, x\in \square(S)\}$. $C_1^+$ is compact and so is also
\begin{equation*}
    D_\delta 
    := \{f - \delta; f\in C_1^+\},\quad \delta \in \left(0, 1/2\right).
\end{equation*}
$D_\delta$ can be covered by $N \in \mathbb{N}$ open balls of radius $\delta$ centered at the functions $h_1,\dots, h_N \in D_\delta$. Hence, for all $f\in C_1^+$,
\begin{equation*}
    \max\{0,h_i\} \leq f \leq \max\{0,h_i\} + 2\delta
\end{equation*}
for some $i\in\{1,\dots, N\}$. The functions $\max\{0,h_i\}$ are in $C_1^+$, hence either zero everywhere or with support on an open set in $\square(S)$. By the previous estimate~\eqref{useful_lower_bound} (open set accessibility) there is for each $\max\{0,h_i\}$ falling into the latter category constants $b_i > 0$ and $n_i\in \mathbb{N}$ such that 
\begin{equation*}
    P^{n_i}\max\{0,h_i\}(x) \geq b_i,\quad x\in \square(S).
\end{equation*}
Since the cover of $D_\delta$ is finite, there are common constants $b > 0$ and $n\in \mathbb{N}$ such that for all $\max\{0,h_i\}$ not zero everywhere,
\begin{equation*}
    P^n \max\{0,h_i\}(x) \geq b,\quad x\in\square(S).
\end{equation*}
Let $f\in C_1^+$. If $\max\{0,h_i\} \neq 0$ it follows from $\max\{0,h_i\} \leq f\leq 1$ that $\min_{x\in \square(S)} P^nf(x) \geq b$ and hence $[P^n f]\leq 1-b$. If $\max\{0,h_i\} = 0$, $f \leq 2 \delta$ yields $[P^n f] \leq 2\delta$. Let $c := \max\{1-b, 2\delta\}$, then $c\in (0,1)$ and
\begin{equation*}
    [P^n f] \leq c,\quad f\in C_1^+.
\end{equation*}
\end{proof}

\begin{lemma}
\label{lemma:A}
There exists positive constants $\theta \in (0,1)$ and $a < +\infty$ such that
\begin{equation*} 
    [P^n f] \leq a\|f\|_{\mathcal{L}}\theta^n,\quad f\in \mathcal{L},
\end{equation*}
and $a$ is independent of $f$.
\end{lemma}

\begin{proof}
Recall that $P$ is Lipschitz friendly with 
\begin{equation*}
    \lambda^+
    = \frac{2\lambda}{1-\beta},\quad \lambda^- = \beta\lambda^+ + \lambda
\end{equation*}
for any $\lambda = \max\{1,m_1m_2L_\sigma\}$ as in the previous section. We notice the following transformation: for all $f\in \mathcal{L}$, there exists an $\bar f_0 \in \mathcal{L}$ non-negative with $\|\bar f_0\|_\infty\leq 1$, and an $y_0\in \mathbb{R}$, such that
\begin{equation*}
    f = [f]\bar f_0 + y_0
\end{equation*}
Denote the Lipschitz constant of $\bar f_0$ by $\bar \lambda_0$. Let $f_0$ be the following scaling of $\bar f_0$:
\begin{equation*}
    f_0 = \frac{\bar f_0}{\max\{1,\bar\lambda_0/\lambda^+\}}.
\end{equation*}
Then $f_0 \in B$ and we have that
\begin{equation*}
    f = \max\{1,\bar\lambda_0/\lambda^+\}[f]f_0 + y_0
\end{equation*}
and hence that
\begin{equation*}
    P^{n+1}f = \max\{1,\bar\lambda_0/\lambda^+\}[f]P^{n+1}f_0 + y_0.
\end{equation*}
Since $Pf_0 \in C_1^+$, by Lemma~\ref{lemma:B} we have that $[P^nf_0]\leq c$ where $c\in (0,1)$. Furthermore, by Lipschitz friendliness of $P$, $\|P^{n+1}f_0\|_{\mathcal{L}} \leq \lambda^+$. Therefore there exists an $\bar f_1 \in \mathcal{L}$ with $\|\bar f_1\|_\infty \leq 1$ and $\|\bar f_1\|_{\mathcal{L}} \leq \lambda^+/c$, and an $\bar y_1 \in \mathbb{R}$, such that
\begin{equation*}
    P^{n+1}f_0 = c \bar f_1 + \bar y_1    
\end{equation*}
Let $p$ be the smallest integer such that
\begin{equation*}
    \left(\frac{\lambda^-}{\lambda^+}\right)^p 
    = \left(\frac{1-\beta}{2}\right)^p < c
\end{equation*}
Since $\|P\bar f_1\|_{\mathcal{L}} \leq \lambda^-/c$ we get that $\|P^p \bar f_1\|_{\mathcal{L}} \leq \lambda^+$ and hence $P^p \bar f_1 \in B$. 

Denote $f_1 := P^p \bar f_1$, then for some $y_1\in \mathbb{R}$, 
\begin{equation*}
    P^{n+p+1} f = \max\{1,\lambda_0/\lambda^+\}[f]cf_1 + y_1.
\end{equation*}
Recursively carrying out the procedure $\ell\in \mathbb{N}$ times, we get that
\begin{equation*}
    P^{\ell(n+p)+1} f 
    = \max\{1,\lambda_0/\lambda^+\}[f]c^\ell f_\ell + y_\ell,\quad f_\ell \in B,\ y_\ell\in \mathbb{R}.
\end{equation*}
We now note that since $\lambda_0[f]$ is the Lipschitz coefficient of $f$ and $\square(S)$ is compact:  
\begin{equation*}
    \max\{1,\lambda_0/\lambda^+\}[f] \leq \|f\|_{\mathcal{L}}
\end{equation*}
Since also $[Pf]\leq [f]$, there exists $a\in \mathbb{R}$ and $\theta\in (0,1)$ independent of $f$ so that
\begin{equation*}
    [P^nf]\leq a\|f\|_{\mathcal{L}}\theta^n,\quad f\in \mathcal{L}.
\end{equation*}
\end{proof}

\subsubsection{Proof of Uniform Ergodicity}
\label{subsubsec:ergodicity}

The existence of a probability measure invariant with respect to $P$ follows readily by the compactness of $\square(S)$ and e.g. Schauder's fixed point theorem. The set $\mathcal{P}(\square(S))$ endowed with the weak topology (metrized by the $1$-Wasserstein distance) is compact, since $\square(S)$ is a compact Euclidean set. As a mapping from $\mathcal{P}(\square(S))$ to itself, $P$ is therefore compact. It is also continuous: let $\mu^n \rightarrow \mu$ in $\mathcal{P}(\square(S))$ and let $f\in \mathcal{L}$. Then by following the calculations of Section~\ref{sec:lipschitz-friendly} we see that $Pf$ is Lipschitz continuous, hence as $n\rightarrow\infty$
\begin{align*}
    \int_{\square(S)} f(x) \mu^n P(dx) 
    &= \int_{\square(S)}Pf(x)\mu^n(dx)  
    \\  
    &\rightarrow \int_{\square(S)} Pf(x)\mu(dx) = \int_{\square(S)} f(x)\mu P(dx)
\end{align*}
implying that $\mu^nP \rightarrow \mu P$. 
The existence of a fixed point now follows by Schauder's theorem.

Towards uniqueness, assume that $\mu',\mu \in \mathcal{P}(\square(S))$ are invariant with respect to $P$. Then, by using the result Lemma~\ref{lemma:A}, we have that for all $f\in \mathcal{L}$, $P^n f \rightarrow K(f)$, a constant depending on $f$. Hence
\begin{equation}
    \mu' f = \lim_{n\rightarrow\infty} \mu' P^n f = K(f) = \lim_{n\rightarrow\infty} \mu P^n f = \mu f,
\end{equation}
and $\nu = \mu$ in the weak topology and the invariant probability measure is unique.

We now turn to the geometric convergence. Let $L_1 := \{f\in \mathcal{L} : \|f\|_\mathcal{L} \leq 1,\ f \text{ bounded }\}$ and let $\|\cdot\|_W$ denote the $1$-Wasserstein metric on $\mathcal{P}(\square(S))$. Let furthermore $\mu$ be the invariant probability measure with respect to $P$. Recall that for all $f\in L_1$, there exists a function $g \in L_1$ bounded by $1$ and an $y \in \mathbb{R}$ such that $P^n f = [P^nf]g + y$, where the brackets $[\ \cdot\ ]$ denotes the spread of a function. By Lemma~\ref{lemma:A} and the Kantorovich-Rubenstein theorem, for all $\nu \in \mathcal{P}(\square(S))$,
\begin{align*}
    \|\nu P^n - \mu\|_W 
    &= \sup_{f\in L_1} | \sigma P^n f - \mu f|
    \\
    &= \sup_{f\in L^1}| \nu P^n f - \mu P^n f|
    \\
    &= \sup_{f\in L_1}[P^n f]\left|\int_{\square(S)} g(x)\nu(dx) - \int_{\square(S)}g(x)\mu(dx)\right|
    \\
    &\leq 2 a \theta^n.
\end{align*}

\subsection{Proof of Theorem~\ref{thm:min_curb}}

\begin{proof} 
The proof consists of four steps.

\textit{Step 1. Bounding the probability of reaching $B_\delta(\mathscr{C})$ in finite time.} \\ 
To find a lower bound for the probability to go from an arbitrary point $p(t) \in B_\delta(\mathscr{C})^C$ to $B_\delta(\mathscr{C})$ in finite time we create a particular path of positive probability that does exactly that. Let $p(t) \in \square (S)$ be given and let 
$s^{(1)} \in S_1\times S_2$ be the strategy profile played at in period $t$. Either $s^{(1)}$ is a CURB block, or the best reply set to $s^{(1)}$ contains a strategy not in $s^{(1)}$, $BR(\overrightarrow{1_{s^{(1)}}}) \not \subset s^{(1)}$. If the former statement is true this step of the proof is complete. That is not always the case, therefore assume that we are in the case of the latter statement, i.e. that the best reply set to $s^{(1)}$ contains a strategy not in $s^{(1)}$. Then, the probability of both players only sampling $s^{(1)}$ at time $t+1$ is bounded from below by
$(1-\beta)^{2k}$. Hence the probability of a strategy profile 
$s^{(2)} \in BR(\overrightarrow{1_{s^{(1)}}}), s^{(2)} \neq s^{(1)}$, being played is bounded from below by
\begin{equation}
    \mathbb{P}\left(\widetilde{BR}(p(t)) = s^{(2)}\ |\ p(t)\right) 
    \geq \frac{(1-\beta)^{2k}}{m_1m_2}(1-\varepsilon)^2.
\end{equation}
Now let $F_2$ be the smallest block $F_2 \in S_1 \times S_2$ that contains $\{s^{(1)}, s^{(2)}\}$. Either $F_2$ is a CURB block or $BR(\Delta(F_2)) \not \subset F_2$, in which case there is at least one sample $D$ of size $k$ from $F_2$ such that $BR(D) \not \subset F_2$. The probability of sampling that particular $D$, and the best replies to $D$ being such that at least one of them is not in $F_2$, is again bounded away from zero. Until we have sampled a sequence of strategy profiles, each extending the set $F_i$, such that $F_i$ is a CURB block, there is always some sample with positive sampling probability such that $BR(D) \not\subset F_i$. The probability of playing a strategy $s^{(i)}$ which is a best reply to $D$ which is not in $F_i$,
${s^{(i)}}\in BR(D)\cap(F_i)^C$, is therefore bounded from below by
\begin{equation}
    \mathbb{P}\left(\widetilde{BR}(p(t+i-1)) = s^{(i)}\ |\ p(t+i-1) \right)
    \geq \frac{\left(\beta^{i-1}(1-\beta)\right)^{2k}}{m_1m_2}(1-\varepsilon)^2.
\end{equation}
Keep filling $F_i, F_{i+1}, F_{i+2},\dots$ with strategies from the CURB block in this fashion, so that $F_T$ spans a CURB block and $T \leq m_1+m_2$ \cite[Lemma 1]{Hurkens1995}. To get a uniform lower bound, assume that $T = m_1 + m_2$ and that once $F_i$ is a CURB block the following $T - i$ strategy profiles are inside the CURB block. 
The probability of this progression of plays is bounded from below: let $\mathcal{E}$ be the event that $p(t+T)$ puts at most $\beta^{T+1}$ mass outside the CURB block spanned by $F_T$, then
\begin{equation}
\label{eq:uniform_lower_bound}
    \mathbb{P}\left( \mathcal{E} \right)
    \geq \frac{\left(\beta^{2k}\right)^{(T-1)!}\left(1-\beta\right)^{2Tk}}{m_1^Tm_2^T}(1-\varepsilon)^{2T}.
\end{equation}
Inside the CURB block spanned by $F_T$, there is a minimal CURB block which we denote by $C = C_1\times C_2$. The probability of both players sampling from $C$ given the state $p(t+T)$ (as described above) is greater or equal to
\begin{equation}
    \mathbb{P}\left((D_1/k, D_2/k) \in \square (C) \ |\ D\text{ from } p(t+T) \right)
    \geq \left(\beta^T(1-\beta)\right)^{2k}(1-\varepsilon)^2.
\end{equation} 
Starting from $p(t)\in B_\delta(C)^C$, a sequence of plays that results in $p(t + T + T^*) \in B_\delta(C)$ is to play $T$ strategies to fill $F_T$ followed by $T^*$ strategies from the minimal CURB block $C$. Conditional on $p(t)\in B_\delta(\mathscr{C})^C$ and the aforementioned event $\mathcal{E}$, the probability that $p(t + T + T^*)\in B_\delta(C) \subset B_\delta( \mathscr{C})$ is bounded from below by
\begin{equation}
    \begin{aligned}
    &\mathbb{P}\Big( (D_1,D_2)(t+T+i) \in \square (C), i=0,\dots, T^*-1\ |\  p(t+T) \text{ as above}\Big)
    \\
    &\geq \left(\beta^T(1-\beta)(1-\varepsilon)\right)^{2kT^*}
     =: \gamma(\varepsilon,T,T^*).
    \end{aligned}
\end{equation}
Now $p(t+T+T^*)$ gives at most $\beta^{T^*}$ probability to all strategy profiles outside $\square(C)$. Therefore, we pick $\delta>0$ and let $T^*\in\mathbb{N}$ be such that $\beta^{T^*} < \delta$ and, summarizing the analysis in this step, we have derived a bound on the probability of moving from any point $p(t) \in B_\delta(\mathscr{C})^C$ to $B_\delta(\mathscr{C})$ in $T + T^\ast$ steps. We denote this bound by $\underline{K}$ and it is given by
 \begin{equation}
    \begin{aligned}
     &P^{T+T^\ast}(p(t), B_\delta(\mathscr{C})) 
     \ \
     &\geq \frac{\left(\beta^{2k}\right)^{(T-1)!}\left(1-\beta\right)^{2TK}(1-\varepsilon)^{2T}}
     {m_1^Tm_2^T} \gamma(\varepsilon,T,T^*) =: \underline{K}.
 \end{aligned}
 \end{equation}
 
\textit{Step 2. Expected exit time from $B_\delta(\mathscr{C})$.} \\ 
Once in $B_\delta(\mathscr{C})$, one of two things must happen for the process to leave. Either one player makes a mistake or one player samples at least one strategy from outside the minimal CURB block $C$ the process is currently centered around. So instead of calculating 
the time to the first exit, denoted $\tau_\varepsilon$, we calculate the expected time until one of these two things happen the first time. Let $\tau_\varepsilon^\ast$ denote the time, starting from $t=0$, until either a strategy is sampled outside $C$ or one player makes an $\varepsilon$-tremble. We denote the expression for the probability that $\tau_\varepsilon^\ast > t^*$, $t^*\in\mathbb{N}$, with $Q_\varepsilon(t^*)$,
\begin{equation}
\label{eq:Qvareps}
    Q_{\varepsilon}(t^*)
    := \mathbb{P}\left(\tau^*_\varepsilon > t^*\ |\ p(0)\in B_{\delta}(C)\right)
    = \prod_{t=0}^{t^*} (1-\beta^t\delta)^{2k}(1-\varepsilon)^2.
\end{equation}
For the case $\varepsilon = 0$, we use the fact that $\sum_{t=0}^\infty \beta^t \delta$ is convergent to conclude that $\prod_{t=0}^\infty (1-\beta^t\delta)^{2k}$ approaches a non-zero limit. Since $Q_\varepsilon$ is decreasing and non-negative,
\begin{equation}
\label{eq:limit_q_star}
    \lim_{t^* \rightarrow \infty} Q_{\varepsilon}(t^*) =
    \begin{cases}
    Q^* \in (0,1),& \text{if }\varepsilon = 0,
    \\
    0,& \text{if }\varepsilon > 0.
    \end{cases}
\end{equation}
We can now derive a bound for $\tau_\varepsilon$, the expected time to exit from $B_\delta(\mathscr{C})$,
\begin{equation}
    \begin{aligned}
    \mathbb{E}\left[\tau_\varepsilon\right] 
    &\geq \mathbb{E}\left[\tau^\ast_\varepsilon\right] 
    \\
    &\geq \mathbb{E}\left[\tau_\varepsilon^\ast\ |\ \tau_\varepsilon^\ast \geq t^*, p(0)\in B_\delta(C) \right]
    \\
    &\qquad \qquad \times \mathbb{P}(\tau_\varepsilon^\ast \geq t^*\ |\
    p(0)\in B_\delta(C))\mathbb{P}(p(0)\in B_\delta(C))) 
    \\
    &\geq t^*Q_{\varepsilon}(t^*)\nu(B_\delta(C)),
    \end{aligned}
\end{equation}
where $\nu$ is the initial distribution of the state process and $\nu(B_\delta(C))$ is the probability that $p(0)\in B_\delta(C)$. We know that the state process converges weakly to the invariant distribution for all initial distributions and therefore $\nu$ is any  distribution on $\square(S)$ of our choice. Choosing $\nu$ as the distribution of the constructed $p(t+T+T^*)$ from above,
\begin{equation}
    \begin{aligned}
    E[\tau_\varepsilon]
    &\geq t^* \prod_{t=0}^{t^*}(1-\beta^t\delta)^{2k}(1-\varepsilon)^2
    \\
    &= t^*(1-\varepsilon)^{2t^*}Q_0(t^*)
    \\
    &\geq t^*(1-\varepsilon)^{2t^*}Q^*,
    \end{aligned}
\end{equation}
where $t^*$ is any positive integer. For a fixed $\varepsilon$, the function $t^* \mapsto t^*(1-\varepsilon)^{2t^*}$ is maximized by 
$t^*(\varepsilon) = - (2\ln(1-\varepsilon))^{-1}$. There is therefore a decreasing sequence of positive numbers $(\varepsilon_j)_{j=1}^\infty$, tending to zero as $j\rightarrow\infty$, such that $t^*(\varepsilon_j)$ is an integer and
\begin{equation}
    \mathbb{E}[\tau_\varepsilon] \geq -\frac{Q^*}{2e\ln(1-\varepsilon_j)},
\end{equation}
which diverges to $\infty$ as $j \rightarrow \infty$.

\textit{Step 3. Bounding $\mu^*_\varepsilon(B_\delta(\mathcal{C})^C)$ from above.}\\
We know that for any $\varepsilon>0$ there exists a unique invariant probability measure $\mu^\ast_\varepsilon$. We also have a lower bound for $P(x, B_\delta(\mathscr{C}))$ uniform over $x\in B_\delta(\mathscr{C})^C$, and a lower bound for the expected time the process stays in $B_\delta(\mathscr{C})$ once it has entered. 

The probability given by the invariant distribution to the set $B_\delta(\mathscr{C})$ is at least the sum over $n$ of the probability of: the state process not being in it $(n+1)(T+T^*)$ steps ago, but in it $n(T+T^*)$ steps ago, and then staying there for at least $n(T+T^*)$ time steps,
\begin{align*}
	1 \geq \mu_\varepsilon^\ast (B_\delta(\mathscr{C})) 
	&\geq \sum_{n=0}^\infty \left(\int_{B_\delta(\mathscr{C})^C}P^{T + T^\ast}\left(x, B_\delta(\mathscr{C})\right) d\mu_\varepsilon^\ast(x) \right) \mathbb{P}\left( \tau_\varepsilon \geq n(T+T^*)\right)   
	\\ 
	&\geq \mu^\ast_\varepsilon\left(B_\delta(\mathscr{C})^C\right) \underline{K}\left( \sum_{n=0}^\infty \mathbb{P}\left( \frac{\tau_\varepsilon}{T+T^*} \geq n \right)\right)
	\\ 
	&\geq \mu^\ast_\varepsilon\left(B_\delta(\mathscr{C})^C\right) \frac{\underline{K}}{T+T^*}\; \mathbb{E}\left[\tau_\varepsilon^\ast \right].
\end{align*}

\textit{Step 4. Putting it all together.}\\
The collection $(\mu^*_\varepsilon)_{\varepsilon>0}$ is tight because $\square(S)$ is compact. So there exists a subsequence that converges weakly to $\mu^* \in \mathcal{P}(\square(S))$. The limit $\mu^*$ is not necessarily unique, however, by the Portmanteau theorem,
\begin{equation}
    \underset{\varepsilon\rightarrow 0}{\lim\inf}\ \mu^*_\varepsilon (U) \geq \mu^*(U) 
\end{equation}
for all open sets $U$ of $\square(S)$. Note that $B_\delta(\mathscr{C})^C$ is open, and
\begin{equation}
    \mu^*_\varepsilon(B_\delta(\mathscr{C})^C)  
    \leq \frac{T + T^*}{\underline{K}\mathbb{E}[\tau^*_\varepsilon]}.
\end{equation}
Since $\underline{K} > 0$ increases as $\varepsilon \to 0$, $\mathbb{E}\left[ \tau_\varepsilon^\ast \right] \to \infty$ as $\varepsilon \to 0$, and $T+T^*$ does not depend on $\varepsilon$,
\begin{equation}
    \mu^*(B_\delta(\mathscr{C})^C) 
    \leq \underset{\varepsilon\rightarrow 0}{\lim\inf}\ \mu^*_\varepsilon(B_\delta(\mathscr{C})^C) 
    \leq (T+T^*)\underset{\varepsilon\rightarrow 0}{\lim\inf}\frac{1}{\underline{K}\mathbb{E}[\tau^*]} 
    = 0.
\end{equation} 
We conclude that that $\mu^\ast_\varepsilon\left(B_\delta(\mathscr{C})\right) \to 1$ as $\varepsilon \to 0$.
\end{proof}

\section{Concentration Around Approximate Nash Equilibrium}
\label{sec:proofs2}

Parts of this appendix relies on the assumption that the game is of size $2\times 2$ and has a unique mixed Nash Equilibrium. Generically, all $2\times 2$ games without pure Nash equilibria must have the basic Matching Pennies structure. One player will be 'agreeing' and the other 'disagreeing' in the sense that the best reply of the agreeing player is to play the same strategy ($0$ or $1$) as the disagreeing player. On the other hand, the disagreeing player's best reply is to not play the same strategy as the agreeing player. Any other situation will generically yield at least one pure equilibrium, and generically a strict pure equilibrium. 

\subsection{Unique Fixed Point to the Expected Best Reply}
\label{app:assumption-7-proof}

\begin{lemma}
\label{lemma:unique_fp}
Let $G$ be a $2\times 2$ game with a unique mixed Nash equilibrium $\hat x$ and let $k$, the number of samples, be an integer such that $\hat x_1 k\not\in\mathbb{N}$ and $\hat x_2 k\not\in\mathbb{N}$. Then there exists a unique fixed point $x^*=(x^*_1, x^*_2) \in \textup{int}(\square ( S ))$ to the system
\begin{equation}
\label{eq:fp}
    \left\{
    \begin{aligned}
    \mathbb{E}\left[\widetilde{BR}_1(x^*_{2}) \right] = x^*_1,
    \\
    \mathbb{E}\left[\widetilde{BR}_2(x^*_{1}) \right] = x^*_2.
    \end{aligned}
    \right.
\end{equation}
\end{lemma}

\begin{proof}
We will refer to the player $1$ and $2$ as the agreeing and the disagreeing player, respectively. The Nash equilibrium $\hat x = (\hat x_1, \hat x_2)$ defines the 'cut-off' $M_i :=  \lfloor \hat x_i k \rfloor$, $i=1,2$. The cut-off is such that if more than $M_1$ of the agreeing player's $k$ samples from the disagreeing player's history are $1$, he plays $1$. The disagreeing player will play strategy $1$ if more than $M_2$ of his $k$ samples from the agreeing player's history of plays are $0$. Consider the function
\begin{equation}
\label{eq:best_rep_is_1}
    \rho_{k,M}(x) 
    := (1-\varepsilon)\sum_{i=M+1}^{k}\binom{k}{i}x^i(1-x)^{k-i} + \varepsilon/2.
\end{equation}
Given that player history is in state $(a,d)$, the probability that the agreeing and disagreeing player plays strategy $1$ is $\rho_a(d) := \rho_{k,M_2}(d)$ and $\rho_d(a) := 1-\rho_{k,M_1}(a)$, respectively. We can now rewrite \eqref{eq:fp} as
\begin{equation*}
    \rho_a(x_2^*) = x_1^*, \qquad \rho_d(x^*_1) = x_2^*.
\end{equation*}
The range of $\rho_a$ and $\rho_b$ is $I_\varepsilon := [\varepsilon/2, 1-\varepsilon/2]$. Therefore, by the strict monotonicity and the continuity of $\rho_a$ and $\rho_d$, we may rewrite \eqref{eq:fp} again, now as
\begin{equation*}
    \begin{aligned}
    \left(\rho_a\circ \rho_d\right) (x_1^*) &= x_1^*,\quad x_1^* \in I_\varepsilon,
    \\
    \left(\rho_d\circ \rho_a\right)(x_2^*) &= x_2^*, \quad x_2^* \in I_\varepsilon.
\end{aligned}
\end{equation*}
Note that since $\rho_a$ and $\rho_d$ are strictly increasing and decreasing, respectively, both $\rho_a\circ \rho_d$ and $\rho_d\circ \rho_a$ are strictly decreasing functions from $[0,1]$ to $[\rho_d(1-\varepsilon/2), \rho_d(\varepsilon/2)]$ and $[\rho_a(\varepsilon/2), \rho_a(1-\varepsilon/2)]$, respectively. Therefore
\begin{equation*}
    \begin{aligned}
    \min\{\rho_a\circ \rho_d (\varepsilon/2), \rho_d\circ \rho_a( \varepsilon/2)\} 
    &\geq \min\{\rho_d(1-\varepsilon/2), \rho_a(\varepsilon/2)\} > \varepsilon/2,
    \\
    \max\{\rho_a\circ \rho_d(1-\varepsilon/2), \rho_d\circ \rho_a(1-\varepsilon/2)\} 
    &\leq \max\{\rho_d(\varepsilon/2), \rho_a(1-\varepsilon/2)\}
    < 1-\varepsilon/2.
    \end{aligned}
\end{equation*}
Hence, since $\rho_a\circ \rho_d$ and $\rho_d\circ \rho_a$ are continuous, they intersect the straight line $x = y$ at a (function-wise) unique point in their respective images and these intersection points are $x^*_1$ and $x^*_2$.
\end{proof}

\subsection{Global Exponential Stability of Mean-Field Dynamics}
\label{app:gas}

Denote by $\xi$ the solution mapping of $\dot{x}(t) = F(x(t))$, $x(0) = p$, where $F(x) := \mathbb{E}[\widetilde{BR}(x)] - x$. Then
\begin{equation}
\label{eq:dy-sys-xi}
    \xi(t,p) = p + \int_0^t F(\xi(s,p))ds.
\end{equation}

\begin{lemma}
\label{lemma:gas}
Let $\Sigma$ contain all points $x\in\square(S)$ such that $F(x) = 0$ or such that $\xi(t,x)$ satisfies $(\xi(t,x)-y)^*F(\xi(t,x)) = 0$ for all $t\geq 0$ and some $y$, such that $F(y)=0$. The mapping $t\mapsto \xi(t,p)$ is globally asymptotically stable, with $\lim_{t\rightarrow \infty} \xi(t,p) \in \Sigma$. Furthermore, if the game is $2\times 2$ with a unique mixed Nash equilibrium, then $\Sigma = \{x^*\}$, the unique root of $F$.
\end{lemma}

\begin{proof}
Let $V(x) := \frac{1}{2}\|x - x^*\|_2^2$ where $x^*$ is a root of $F$. The existence of $x^*$ is granted by Brouwer's fixed point theorem; $\square(S)$ is compact and convex and $F$ is continuous. Differentiating $V$ with respect to time at the solution mapping $\xi(t,p)$, we get
\begin{equation}
    \begin{aligned}
    -\dot{V}(\xi(t,p))
    &= -\nabla V(\xi(t,p))\dot{\xi}(t,p)
    \\
    &= -( \xi(t,p)-x^*)^TF(\xi(t,p))
    \\
    &= -(\xi(t,p)-x^*)^T\left(\mathbb{E}[\widetilde{BR}(\xi(t,p))\ |\ \xi(t,p)] - \xi(t,p)\right)
    \\
    &= 2V(\xi(t,p)) - (\xi(t,p)-x^*)^T
    \left(\mathbb{E}[\widetilde{BR}(\xi(t,p)\ |\ \xi(t,p)] - x^*) \right)
    \\
    &= V(\xi(t,p)) -  V(\mathbb{E}[\widetilde{BR}(\xi(t,p))\ |\ \xi(t,p)]) 
    \\
    &\qquad + \frac{1}{2}\|\xi(t,p) - \mathbb{E}[\widetilde{BR}(\xi(t,p))\ |\ \xi(t,p)]\|^2_2
    \end{aligned}
\end{equation}
where in the last step we used the identity $2y^Tz = \|y\|_2^2 + \|z\|_2^2 - \|y-z\|_2^2$, $y,z\in\mathbb{R}^d$. We notice that
\begin{equation}
    \begin{aligned}
    &V(\mathbb{E}[\widetilde{BR}(\xi(t,p))\| \xi(t,p)])
    \\
    &= \frac{1}{2}\|\mathbb{E}[\widetilde{BR}(\xi(t,p))\ |\ \xi(t,p)] - \xi(t,p) + \xi(t,p) - x^*\|^2_2
    \\
    &\leq \frac{1}{2}\|\mathbb{E}[\widetilde{BR}(\xi(t,p))\ |\ \xi(t,p)] - \xi(t,p)\|^2_2 + V(\xi(t,p)),
    \end{aligned}
\end{equation}
hence $\dot{V}(\xi(t,p)) \leq 0$. Furthermore, $V$ is radially unbounded. Let $R := \{x \in \square(S) : (x-x^*)^TF(x) = 0\}$, then $R=\{x\in\square(S):\dot{V}(x)=0\}$ and $R$ contains $x^*$, any other point solution to $F(x) = 0$, and all $x$ such that the vectors $(x-x^*)$ and $F(x)$ are orthogonal. By a global invariant set theorem \cite[Thm. 3.5]{slotine1991applied}, $\xi(t,p)$ converges to the largest invariant set of $R$, which is $\Sigma$.

Next, for $2\times 2$ games with a unique mixed Nash equilibrium, we show the points in $R$ different from $x^*$ (now unique) cannot be in $\Sigma$. First note that if $x\in R\backslash \{x^*\}$, then $x_i \neq x^*_i, i=1,2$. Without loss of generality, assume that player 2 has the disagreeing role and that $x_0 > x^*$. If $x_0\in R\backslash\{x^*\}$ then $F(x_0) \neq 0$ and a trajectory starting in $x_0$ will evolve according to the dynamic system $\dot{x}(t) = F(x(t)),\, x(0) = x_0$. Assume, towards a contradiction, that $x(t)\in R\backslash \{x^*\}$ for all $t\geq 0$. After some finite positive time, call it $t^*$, the path must cross the line $(x,x^*_2; x\in[0,1])$ (because the trajectory starts at at $x_0 > x^*$ and player 2 is disagreeing, it will move "south-east" in $\square (S)$). This crossing contradicts $x(t^*) \in R\backslash\{x^*\}$ since $x(t^*) \in R\backslash\{x^*\}$ would require both components of $x(t^*)$ to be different from $x^*$. The same argument can be carried out for all other possible initial positions ($x_0 - x^* < 0$ or mixed signs) and for switched player roles. It follows that $\{x^*\}$ is the only invariant set in $R$. 
\end{proof}

\subsection{Trajectories over Bounded Time Intervals}

By \cite[Lemma 1]{benaim2003deterministic}, the state process $p(\cdot)$ and its mean-field approximation $\xi(\cdot,p(0))$ lie close to each other (over bounded time intervals) with high probability. We have to do one modification to apply the result: we re-scale size of the time steps taken by our learning process. This has no effect on previous results since we will always (for a fixed $\beta$) have a fixed positive step size. The original proof of \cite{benaim2003deterministic} can be used to prove the lemma below.

\begin{lemma}
\label{lemma:bdd-time-intervals}
Scale the step size of $t$ by $(1-\beta)$. Let $T = N(1-\beta)$ for some $\mathbb{N}\in\mathbb{N}$ and let $(\hat{p}(t); t \in [0,T])$ be the linear interpolation of the path $(p(t); t = 0, 1-\beta,\dots,(1-\beta)N)$. Then, for all $\eta > 0$,
\begin{equation}
    \mathbb{P}\left(\max_{t\in[0,T]}\|\hat{p}(t) - \xi(t,p(0))\|_\infty \geq \eta \right) \leq 2(m_1+m_2-2)e^{-\eta^2 c}
\end{equation}
where $c$ is a positive constant and proportional to $e^{-\gamma T}(T(1-\beta))^{-1}$, where $\gamma > 0$ depends only on the size of the game.
\end{lemma}

\subsection{Proof of Theorem~\ref{thm:concetration_on_Nash}}

Let $t\geq s \geq 0$. Below, $K$ will denote a generic positive constant. Whenever $\eta > \|\xi(t,\hat{p}(t-s)) - \xi(t,p(0))\|_\infty$, Lemma~\ref{lemma:bdd-time-intervals} yields that
\begin{equation}
    \begin{aligned}
    &\mathbb{P}(\|\hat{p}(t) - \xi(t,0)\|_\infty \geq \eta)
    \\
    &\leq \mathbb{P}\left(\|\hat{p}(t) - \xi(t,\hat{p}(t-s))\|_\infty \geq \eta-\|\xi(t,\hat{p}(t-s) - \xi(t,p(0))\|_\infty\right)
    \\
    &\leq K\exp\left(-(\eta - \|\xi(t,\hat{p}(t-s)) - \xi(t,p(0))\|_\infty)^2K\frac{e^{-\gamma s}}{s(1-\beta)} \right).
    \end{aligned}
\end{equation}
Furthermore,
\begin{equation}
    \begin{aligned}
    &\mathbb{P}(\|\hat{p}(t) - x^*\|_\infty \geq \eta)
    \\
    &\leq \mathbb{P}\left(\|\hat{p}(t) - \xi(t,p(0))\|_\infty \geq \eta-\|\xi(t, p(0)) - x^*\|_\infty\right),
    \end{aligned}
\end{equation}
so we have that
\begin{equation}
    \begin{aligned}
    &\mathbb{P}(\|\hat{p}(t) - x^*\|_\infty \geq \eta) 
    \leq \mathbb{P}\Big(\|\hat{p}(t) - \xi(t,\hat{p}(t-s))\|_\infty \geq \eta
    \\
    &-\|\xi(t,\hat{p}(t-s)) - \xi(t,p(0))\|_\infty - \|\xi(t, p(0)) - x^*\|_\infty\Big)
    \\
    &\leq K\exp\Bigg( -(\eta - \|\xi(t,\hat{p}(t-s)) - \xi(t,p(0))\|_\infty - \|\xi(t, p(0)) - x^*\|_\infty)^2
    \\
    &\hspace{3cm} \times K\frac{e^{-\gamma s}}{s(1-\beta)}\Bigg).
    \end{aligned}
\end{equation}
Letting $t\rightarrow \infty$, we know from Lemma~\ref{lemma:gas} that $\xi(t,p(0))\rightarrow x^*$, so
\begin{equation}
\begin{aligned}
    &\lim_{t\rightarrow\infty}\mathbb{P}(\|\hat{p}(t) - x^*\|_\infty \geq \eta) 
    \\
    &\leq \sup_{x\in \square(S)}K\exp\left(-(\eta - \|\xi(s,x) - x^*\|_\infty)^2 K\frac{e^{-\gamma s}}{s(1-\beta)} \right).
\end{aligned}
\end{equation}
Now for $\sigma$ large enough it holds that $\|\xi(s,x)-x^*\|_\infty \leq \eta/2$ uniformly in $x$ for all $s\geq \sigma$. Thus
\begin{equation}
    \lim_{t\rightarrow\infty}\mathbb{P}\left(\|\hat{p}(t) - x^*\|_\infty^2\geq \eta \right) 
    = o\left(\exp\left(-\frac{K\eta^2}{1-\beta}\right)\right),
\end{equation}
proving the theorem.

\section{Extension of Theorem~\ref{thm:concetration_on_Nash}}
\label{sec:extension}

This section present a strategy to prove the extension which was hinted at in the end of Section~\ref{sec:insideCURB}. It can be summarized as follows. Conditioning on the event that that the game has behaved like the subgame on a minimal CURB for the last rounds, the probability of the process staying close to the Nash equilibrium in that particular block can be controlled with $\beta$. This is then combined with concepts from the proof of Theorem~\ref{thm:min_curb} to bound the asymptotic probability of the process being in the neighbourhood of a Nash equilibrium. The argument relies a lemma with the same assumptions as Theorem~\ref{thm:concetration_on_Nash}, preventing further generalization to games with a general minimal CURB configuration.

Take for some $j$ any point $\rho$ in $B_\delta(C_j)$ and let $p(\cdot)$ be the RWS state process with $p(t) = \rho$.  Denote by $\tau^*_\varepsilon(t)$ the number of games played since either a strategy outside $C_j$ was sampled or an $\varepsilon$-tremble occurred, counting backward from $t$. Hence, in the $\tau^*_\varepsilon(t)$ periods $\{t-\tau^*_\varepsilon(t),\dots, t\}$ the
state process has behaved as a process with $\varepsilon = 0$ on the subgame given by $C_j$. Let furthermore $\check x(\cdot)$ be the mean process constrained at $t-\tau^*_\varepsilon(t)$:
\begin{equation}
    \begin{aligned}
    &\check x(t+1)
    = \beta \check x(t) + (1-\beta)\mathbb{E}\left[\widetilde{BR}(\check x(t))\right], \quad t \geq \tau^*_\varepsilon(t),
    \\
    &\check x(t-\tau^*_\varepsilon(t)) = p(t-\tau^*_\varepsilon(t)).
    \end{aligned}
\end{equation}
Let $x^*$ denote the Nash equilibrium in $C_j$ and define the event
\begin{equation*}
    \Theta_t(\eta) := \{\|p(t) - x^*\| \geq \eta + \|\check x(t) - x^*\|\},\qquad \eta > 0,\ t\in \mathbb{N}.
\end{equation*} 
The probability of $\Theta_t(\eta)$ can be written as
\begin{equation}
\label{eq:the_last_label}
    \begin{aligned}
    \mathbb{P}(\Theta_t(\eta))
    &= \sum_{r=0}^R \mathbb{P}\left(\Theta_t(\eta) \ |\ \tau^*_\varepsilon(t) = r \right)
    \\
    & \times
    \Big[
    \mathbb{P}\left(\tau^*_\varepsilon(t) = r\ |\ p(t-r)\in B_\delta(\mathscr{C})\right) \mathbb{P}\left(p(t-r)\in B_\delta(\mathscr{C})\right)
    \\
    &+ \mathbb{P}\left(\tau^*_\varepsilon(t) = r\ |\ p(t-r)\in B_\delta(\mathscr{C})^c\right) \mathbb{P}\left( p(t-r)\in B_\delta(\mathscr{C})^c\right)
    \Big]
    \\
    &+ \mathbb{P}\left(\Theta_t(\eta) \ |\ \tau^*_\varepsilon(t) > R\right) \mathbb{P}\left(\tau^*_\varepsilon(t) > R\right).
    \end{aligned}
\end{equation}

Taking the limit $t\rightarrow\infty$ of the first term of \eqref{eq:the_last_label} (the sum), we get
\begin{equation}
    \begin{aligned}
    &\sum_{r=0}^R\mathbb{P}\left(\bar\Theta_r(\eta) \ |\ \bar \tau^*_\varepsilon = r\right)\mathbb{P}\left(\bar \tau^*_\varepsilon = r\ |\ \bar p(0)\in B_\delta(\mathscr{C})\right)
    \mu^*_\varepsilon\left( B_\delta(\mathscr{C})\right)
    \\
    &+ \sum_{r=0}^R\mathbb{P}\left(\bar\Theta_r(\eta) \ |\ \bar \tau^*_\varepsilon = r\right) \mathbb{P}\left(\bar\tau^*_\varepsilon = r\ |\ \bar p(0)\in B_\delta(\mathscr{C})^c\right) \mu^*_\varepsilon\left( B_\delta(\mathscr{C})^c\right)
    \end{aligned}
\end{equation}
where $\bar p$ is an auxilliary RWS starting $\bar p (0) \sim \mu^*_\varepsilon$, $\bar\Theta_r(\eta)$ is the event 
\begin{equation*}
    \bar\Theta_r(\eta) 
    := \{\|\bar p(r) - x^*\| \geq \eta \},\qquad \eta > 0,\ r\in \mathbb{N},
\end{equation*}
and $\bar\tau^*_\varepsilon$ is the number of games played in the RWS $\bar p$, until either a strategy is sampled outside $C_j$ or one player makes an $\varepsilon$-tremble.\footnote{When identifying the set $\bar\Theta_r(\eta)$, we use that the mean process tends to the Nash equilibrium as $t\rightarrow\infty$ under the active set of hypotheses. The notation $\bar\tau^*_\varepsilon$ has a non-bared counterpart in the proof of Theorem~\ref{thm:min_curb}. The bar has been added to emphasize the connection to $\bar p$.} We continue by analyzing the first term. The conditioning on $\bar\tau^*_\varepsilon = r$ means that only the game has behaved like the subgame corresponding to the minimal CURB block for the last $r$ rounds. The next lemma gives an estimate similar to that of Lemma~\ref{lemma:bdd-time-intervals}, but pointwise in time. Its proof is found in Appendix~\ref{sec:last-sec}.

\begin{lemma}
\label{lem:variance}
For all $r\in\mathbb{N}$ and $\eta > 0$
\begin{equation}
\label{eq:variance}
    \mathbb{P}\left(\|\bar p(r)- \bar x(t)\|_2 \geq \eta \right) 
    \leq C\frac{\beta\left(\beta^t + \left( 1-\beta \right)\right)}{\eta^2},
\end{equation}
where $\bar x$ is the mean process corresponding to $\bar p$, constrained at $r = 0$ as $\bar x (0) = \bar p (0)$, and $C$ is a positive constant that depends only on $k$.
\end{lemma}

We use the triangle inequality and Lemma~\ref{lem:variance} to get
\begin{equation}
    \mathbb{P}\left(\bar\Theta_r(\eta) \ |\ \bar\tau^*_\varepsilon = r \right)
    \leq C\frac{\beta\left(\beta^r + (1-\beta)\right)}{\eta^2}.
\end{equation}
Next, we compute the probability of the process behaving as if only the minimal CURB subgame was played for the last $r$ interactions. We take inspiration from the proof of Theorem~\ref{thm:min_curb}. With $Q_\varepsilon$ the function defined in~\eqref{eq:Qvareps} (with $\bar\tau^*_\varepsilon$ taking the role as the random variable), we have that
\begin{equation}
    \begin{aligned}
    &\mathbb{P}\left(\bar\tau^*_\varepsilon = r
    \ |\ \bar p(0) \in B_\delta(\mathscr{C})\right) 
    \\
    &= \mathbb{P}\left(\bar\tau^*_\varepsilon > r-1
    \ |\ \bar p(0) \in B_\delta(\mathscr{C})\right) - \mathbb{P}\left(\bar\tau^*_\varepsilon > r
    \ |\ \bar p(0) \in B_\delta(\mathscr{C})\right)
    \\
    &= Q_\varepsilon(r-1) - Q_\varepsilon(r)
    \\
    &= \left(1
    - (1-\beta^r\delta)^{2k}(1-\varepsilon)^2\right)\prod_{t=0}^{r-1} (1-\beta^t\delta)^{2k}(1-\varepsilon)^2
    \\
    &= \left(1
    -(1-\beta^r\delta)^{2k}(1-\varepsilon)^2 \right)
    (1-\varepsilon)^{2(r-1)}Q_o(r-1)
    \\
    &\leq \left(\frac{1 - (1-\varepsilon)^2}{(1-\varepsilon)^2} + 2k\beta^r\delta \right) (1-\varepsilon)^{2r}.
    \end{aligned}
\end{equation}
The last inequality follows from a closer examination of the remainder in the first order Taylor expansion of $(1-\beta^r\delta)^{2k}$ around $\beta^r\delta = 0$:
\begin{equation}
    R_1(\beta^r\delta) = \frac{2k(2k-1)(1-c)^{2k-2}}{2!}(\beta^r\delta)^2
\end{equation}
for some $c\in [0,\beta^r\delta]$. The remainder is non-negative for all admissible values of $\beta$, $\delta$, and $c$. Furthermore, when $\varepsilon$ is small enough
\begin{equation}
    \mathbb{P}(\bar\tau^*_\varepsilon = r \ |\ \bar p(0)\in B_\delta(\mathscr{C}))
    \leq (3\varepsilon + 2k\beta^r\delta)(1-\varepsilon)^{2r}
\end{equation}
Combining the estimates above, we get for small values of $\varepsilon$
\begin{equation}
\begin{aligned}
    &\sum_{r=0}^R 
    \mathbb{P}\left(\bar\Theta_r(\eta)\ |\ \bar\tau^*_\varepsilon = r\right)
    \mathbb{P}\left(\bar\tau^*_\varepsilon = r\ |\ \bar p(0)\in B_\delta(\mathscr{C})\right)
    \\
    &\leq\sum_{r=0}^R C\frac{\beta\left(\beta^r + 1-\beta\right)}{\eta^2}
    (1-\varepsilon)^{2r}
    \left(3\varepsilon + 2k\beta^r\delta \right)
    \\
    &= C\frac{\beta}{\eta^2}
    \Bigg\{3\varepsilon\sum_{r=0}^R
    \left(\beta^r + 1-\beta\right)(1-\varepsilon)^{2r}
    \\
    &\hspace{1.5cm} + 
    2k\delta\sum_{r=0}^R \left(\beta^r +1-\beta\right)\beta^{r}(1-\varepsilon)^{2r}
    \Bigg\}.
\end{aligned}
\end{equation}
The series are convergent, letting $R\rightarrow\infty$ we get
\begin{equation}
\begin{aligned}
    &\lim_{R\rightarrow\infty}\sum_{r=0}^R 
    \mathbb{P}\left(\bar\Theta_r(\eta)\ |\ \bar\tau^*_\varepsilon = r\right)
    \mathbb{P}\left(\bar\tau^*_\varepsilon = r\ |\ \bar p(0)\in B_\delta(\mathscr{C})\right)
    \\
    &\qquad \leq
    C\frac{\beta}{\eta^2}
    \Bigg\{
    3\varepsilon\left(
    \frac{1}{1 - (1-\varepsilon)^{2}\beta}
    +
    \frac{(1-\beta)}{1-(1-\varepsilon)^2}\right)
    \\
    &\qquad \qquad + 
    2k\delta\left(\frac{1}{1-(1-\varepsilon)^2\beta^2} 
    + \frac{(1-\beta)}{1-(1-\varepsilon)^2\beta}\right)
    \Bigg\}.
\end{aligned}
\end{equation}
Sending $\varepsilon \rightarrow 0$, we get
\begin{equation}
\begin{aligned}
    &\lim_{\varepsilon\rightarrow 0}\lim_{R\rightarrow\infty}\sum_{r=0}^R 
    \mathbb{P}\left(\bar\Theta_r(\eta)\ |\ \bar\tau^*_\varepsilon = r\right)
    \mathbb{P}\left(\bar\tau^*_\varepsilon = r\ |\ \bar p(0)\in B_\delta(\mathscr{C})\right)
    \\
    &\qquad \leq C\frac{\beta}{\eta^2}
    \Bigg\{
    \frac{3}{2}(1-\beta) +  
    2k\delta\left(\frac{1}{1-\beta^2} + 1\right)
    \Bigg\}
    \\
    &\qquad \leq C\frac{\beta}{\eta^2}
    \Bigg\{
    \frac{3}{2}(1-\beta) + 
    2k\delta\left(\frac{1}{1-\beta} + 1\right)
    \Bigg\}.
\end{aligned}
\end{equation}
Finally, choosing $\delta = (1-\beta)^2/8k$, we get
\begin{equation}
\begin{aligned}
    &\lim_{\varepsilon\rightarrow 0}\lim_{R\rightarrow\infty}\sum_{r=0}^R 
    \mathbb{P}\left(\bar\Theta_r(\eta) \ |\ \bar\tau^*_\varepsilon = r\right)
    \mathbb{P}\left(\bar\tau^*_\varepsilon = r\ |\ \bar p(0)\in B_{(1-\beta)^2/8k}(\mathscr{C})\right)
    \\
    &\hspace{2cm} \leq C\frac{4\beta}{\eta^2}(1-\beta).
\end{aligned}
\end{equation}
The term $\mathbb{P}(\bar\tau^*_\varepsilon > R) = 0$ vanishes as $R\rightarrow\infty$ by the same 
analysis that showed $Q^* = 0$ when $\varepsilon \neq 0$, cf. \eqref{eq:limit_q_star}.
Going all the way back to \eqref{eq:the_last_label} and plugging in the derived estimates yields
\begin{equation}
    \lim_{\varepsilon\rightarrow 0}\mu^*_\varepsilon(\Theta_\infty(\eta))
    \leq C\frac{4\beta}{\eta^2}(1-\beta),
\end{equation}
and hence
\begin{equation}
    \lim_{\varepsilon \to 0}\
    \mu^*_\varepsilon \left(\{p\in \square(S) \ |\ \min_{j}\|p - x_j^*\| \geq \eta\} \right) 
    \geq O\left(\frac{1-\beta}{\eta^2}\right).
\end{equation}
The linear estimate we find here is weaker than the exponential estimate of Theorem~\ref{thm:concetration_on_Nash}. However, the concentration of the learning process is controlled in a similar fashion: as $\beta \rightarrow 1$, the RWS concentrates on Nash equilibria.

\subsection{Proof of Lemma~\ref{lem:variance}}
\label{sec:last-sec}

Let $M_1$ and $M_2$ be the cut-offs defined by the unique 
Nash equilibrium $N^*$ as in Lemma~\ref{lemma:unique_fp}, let $x^*$ be the fixed point \eqref{eq:fp}, and let
\begin{equation*}
    p_{k,M}(x) := (1-\varepsilon)\sum_{i=M+1}^{k}\binom{k}{i}x^i(1-x)^{k-i} + \varepsilon/2.
\end{equation*}
Consider $A^*$, the 'shifted' state space
\begin{equation*}
    A^* = (A^*_1, A^*_2) := [-x_1^*, 1-x_1^*]\times[-x_2^*,1-x_2^*].
\end{equation*}
Extend the function $d \mapsto p_{k,M_1}(x_2^* + d)$ to $A^*_2/\beta$ and 
$a \mapsto p_{k,M_2}(x_1^* + a)$ to $A^*_1/\beta$ by keeping the same expression, and define the functions $g_d$ and $g_a$ over $A^*_1$ and $A^*_2$, respectively, by
\begin{equation*}
    \begin{aligned}
    g_d(a) &:= \int_0^a \Big(p_{k,M_2}(x_1^* + z/\beta) - (1-x_2^*) \Big)dz,
    \\
    g_a(d) &:= \int_0^d \Big(p_{k,M_1}(x_2^* + z/\beta) - x_1^* \Big)dz.
    \end{aligned}
\end{equation*}
The functions $g_a$ and $g_d$ are smooth for all $M_1, M_2$ and all $k<\infty$, and $g_a(0) = g_d(0) = 0$. Furthermore, $g_a^\prime (d) = p_{k,M_1}(x_2^* + d/\beta)-x_1^*$ and likewise differentiation of $g_d$ yields the integrand evaluated in the argument. Hence, since $p_{k,M_1}$ and $p_{k,M_2}$ are strictly increasing, $g_a$ and $g_d$ are strictly convex. 
We will make use the following estimates of $g_a$ and $g_d$: there exists four (in general different from each other) positive constants $c_{a-}$, $c_{a+}$, $c_{d-}$, and $c_{d+}$ such that 
\begin{equation*}
    \begin{aligned}
    &p_{k,M_1}(x_2^* + y) \geq c_{a+}y + x_1^*,& &y\in[0,1-x_2^*],
    \\
    &p_{k,M_1}(x_2^* + y) \leq c_{a-}y + x_1^*,& &y\in[-x_2^*,0],
    \\
    &p_{k,M_2}(x_1^* + y) \geq c_{d+}y + (1-x_2^*),& &y\in[0, 1-x_1^*],
    \\
    &p_{k,M_2}(x_1^* + y) \leq c_{d-}y + (1-x_2^*),& &y\in[-x_1^*,0].
    \end{aligned}
\end{equation*}
The estimates imply that for all $(a,d)\in A^*$,
\begin{equation}
\label{eq:g_estimate}
    g_a(d) + g_d(a) \geq \min\{c_{a+},c_{a-}, c_{d+}, c_{d-}\}\frac{1}{2\beta}\left(a^2+d^2\right).
\end{equation}
Now consider the shifted states $A(t) := p_1(t)-n_1^*$ and $B(t) := p_2(t)-n_2^*$. The update of the shifted state is
\begin{equation*}
    \begin{aligned}
        &A(t+1) = \beta A(t) + (1-\beta)\left(\widetilde{BR}_1(x_2^* + B(t))-x_1^*\right),& 
        &A(0) = p_1(0) - x_1^*,
        \\
        &B(t+1) = \beta B(t) + (1-\beta)\left(\widetilde{BR}_2(x_1^* + A(t)) - x_2^*\right),& 
        &B(0) = p_2(0) - x_2^*.
    \end{aligned}
\end{equation*}
For $(a,d) \in A^*$, let $G(a,d) := g_d(a) + g_a(d)$. Expanding $G$ with the Taylor formula yields%
\begin{equation*}
    \begin{aligned}
    &G\left( A\left( t+1\right), B\left( t+1\right) \right) 
    \\
    &= G\left( \beta A\left( t\right) ,\beta B\left( t\right) \right) 
    +
    g^{\prime }_d\left( \beta A\left( t\right) \right) \left( 1-\beta \right)\left(\widetilde{BR}_1(x_2^* + B(t)) - x_1^*\right)
    \\
    &\hspace{15pt}+ g^{\prime }_a\left( \beta B\left( t\right) \right) \left( 1-\beta \right) \left(\widetilde{BR}_2(x_1^* + A(t)) - x_2^*\right) + O\left( \left( 1-\beta \right) ^{2}\right)
    \\
    &= G\left( \beta A\left( t\right) ,\beta B\left( t\right) \right) 
    \\
    &\hspace{15pt} +\left(1-\beta\right)\Bigg(\left(p_{k,M_2}(x_1^* + A(t)) - (1-x_2^*)\right)\left(\widetilde{BR}_1(x_2^* + B(t)) - x_1^*\right)
    \\
    &\hspace{15pt} + \left(p_{k,M_1}(x_2^* + B(t)) - x_1^*\right) \left(\widetilde{BR}_2(x_1^*A(t))-x_2^*\right)\Bigg) + O\left( \left( 1-\beta \right) ^{2}\right),
    \end{aligned}
\end{equation*}
where the ordo is uniform since $g_a$ and $g_d$ are $C^{2}$ over the compact state space. However 
\begin{equation}
    \begin{aligned}
    \mathbb{E}\left[ \widetilde{BR}_1(x_2^* + B(t))\ |\ \mathcal{F}_t\right] &= p_{k,M_1}(x_2^* + B(t)),
    \\
    \mathbb{E}\left[ \widetilde{BR}_2(x_1^* + A(t))\ |\ \mathcal{F}_t\right] &= 1 - p_{k,M_2}(x_1^* + A(t)).
    \end{aligned}
\end{equation} 
Therefore the line of order $1$ has conditional expectation zero and we are left with
\begin{equation}
    \mathbb{E}\left[ G\left( A\left( t+1\right), B\left( t+1\right) \right)\ |\ \mathcal{F}_t \right] 
    \leq G\left( \beta A\left( t\right),\beta B\left( t\right) \right) + M\left( 1-\beta \right)^{2} 
\end{equation}
for some uniform constant $M$. By convexity of $g_a$ and $g_d$,
\begin{equation}
    \mathbb{E}\left[ G\left( A\left( t+1\right), B\left( t+1\right) \right)\
    |\ \mathcal{F}_t\right] 
    \leq \beta G\left( A\left( t\right), B\left( t\right) \right)
    + M\left( 1-\beta \right)^{2} 
\end{equation}
By repeated use of the argument above together with the tower property of 
conditional expectations we get
\begin{equation}
    \begin{aligned}
    & \mathbb{E}\left[ G\left( A\left( t+1\right), B\left( t+1\right) \right)\
    |\ \mathcal{F}_t\right] 
    \\&\qquad \leq \beta^{\tau+1} G\left( A\left( 0\right), B\left( 0\right) \right) + \sum_{\tau=0}^t \beta^\tau M\left( 1-\beta \right) ^{2}.
    \end{aligned}
\end{equation}
So when $t\rightarrow \infty$,
\begin{equation}
\label{eq:G_estimate}
    \underset{t\rightarrow \infty }{\lim \sup }\ \mathbb{E}\left[ G\left( A\left( t\right), B\left( t\right) \right) \right]  
    \leq M\left(1-\beta \right).
\end{equation}
From \eqref{eq:g_estimate} and \eqref{eq:G_estimate} it follows that
\begin{equation}
    \underset{t\rightarrow \infty }{\lim \sup }\ \mathbb{E}\left[ A^{2}\left( t\right) + B^{2}\left( t\right) \right] 
    \leq \frac{2\beta M}{\min\{c_{a+},c_{a-},c_{d+},c_{d-}\}} \left( 1-\beta \right).
\end{equation}
The proof is completed by simply noting that the variance process,
\begin{equation}
    (v_1(t), v_2(t)) := \left(\mathbb{V}(p_1(t)), \mathbb{V}(p_2(t))\right).
\end{equation}
satisfies
\begin{equation}
    v_1(t) + v_2(t) 
    \leq \mathbb{E}\left[\left(p_1(t)- n_1^*\right)^2 
    + \left(p_2(t)-n_2^*\right)^2\right] = \mathbb{E}[A^2(t) + B^2(t)].
\end{equation}
The claim now follows by Chebyshev's inequality.

\bibliographystyle{aea}
\bibliography{references}

@book{villani2008optimal,
  title={Optimal transport: old and new},
  author={Villani, C{\'e}dric},
  volume={338},
  year={2008},
  publisher={Springer Science \& Business Media}
}

@article{PeytonYoung2006,
author = {{Young}, Peyton Hobart and Foster, Dean P.},
journal = {Theoretical Economics},
number = {3},
pages = {341--367},
title = {{Regret testing: learning to play Nash equilibrium without knowing you have an opponent}},
url = {https://econtheory.org/ojs/index.php/te/issue/view/3},
volume = {1},
year = {2006}
}

@article{Fudenberg1993,
author = {Fudenberg, Drew and Kreps, David M.},
journal = {Games and Economic Behavior},
pages = {320--367},
title = {{Learning Mixed Equilibria}},
volume = {5},
year = {1993}
}

@article{Kreindler2013,
author = {Kreindler, Gabriel E. and Young, H. Peyton},
doi = {10.1016/j.geb.2013.02.004},
isbn = {0001409107},
issn = {08998256},
journal = {Games and Economic Behavior},
pages = {39--67},
publisher = {Elsevier Inc.},
title = {{Fast convergence in evolutionary equilibrium selection}},
url = {http://dx.doi.org/10.1016/j.geb.2013.02.004},
volume = {80},
year = {2013}
}

@article{Hurkens1995,
author = {Hurkens, Sjaak},
doi = {10.1006/game.1995.1053},
issn = {08998256},
journal = {Games and Economic Behavior},
month = {nov},
number = {2},
pages = {304--329},
title = {{Learning by Forgetful Players}},
url = {http://linkinghub.elsevier.com/retrieve/doi/10.1006/game.1995.1053 http://linkinghub.elsevier.com/retrieve/pii/S0899825685710536},
volume = {11},
year = {1995}
}

@article{Benam1999,
author = {Bena\"{i}m, Michel and Hirsch, Morris W},
doi = {10.1006/game.1999.0717},
issn = {08998256},
journal = {Games and Economic Behavior},
month = {oct},
number = {1-2},
pages = {36--72},
title = {{Mixed Equilibria and Dynamical Systems Arising from Fictitious Play in Perturbed Games}},
url = {http://linkinghub.elsevier.com/retrieve/pii/S0899825699907170},
volume = {29},
year = {1999}
}

@article{Young1993,
  title={The evolution of conventions},
  author={Young, H Peyton},
  journal={Econometrica: Journal of the Econometric Society},
  pages={57--84},
  year={1993},
  publisher={JSTOR}
}

@article{Hart2006,
author = {Hart, Sergiu and Mas-Colell, Andreu},
doi = {10.1016/j.geb.2005.09.007},
isbn = {9810534124},
issn = {08998256},
journal = {Games and Economic Behavior},
number = {2},
pages = {286--303},
title = {{Stochastic uncoupled dynamics and Nash equilibrium}},
volume = {57},
year = {2006}
}

@article{Foster2003,
author = {Foster, Dean P. and Young, H. Peyton},
doi = {10.1016/S0899-8256(03)00025-3},
issn = {08998256},
journal = {Games and Economic Behavior},
number = {1},
pages = {73--96},
title = {{Learning, hypothesis testing, and Nash equilibrium}},
volume = {45},
year = {2003}
}

@article{Hofbauer2002,
author = {Hofbauer, Josef and Sandholm, William H.},
doi = {10.1111/j.1468-0262.2002.00440.x},
issn = {0012-9682},
journal = {Econometrica},
month = {nov},
number = {6},
pages = {2265--2294},
title = {{On the Global Convergence of Stochastic Fictitious Play}},
url = {http://doi.wiley.com/10.1111/j.1468-0262.2002.00440.x},
volume = {70},
year = {2002}
}

@article{Block2019learning,
  title={Learning dynamics with social comparisons and limited memory 1},
  author={Block, Juan I and Fudenberg, Drew and Levine, David K},
  journal={Theoretical Economics},
  volume={14},
  number={1},
  pages={135--172},
  year={2019},
  publisher={Wiley Online Library}
}

@article{Balkenborg2013,
author = {Balkenborg, Dieter and Hofbauer, Josef and Kuzmics, Christoph},
doi = {10.3982/TE652},
issn = {19336837},
journal = {Theoretical Economics},
month = {jan},
number = {1},
pages = {165--192},
title = {{Refined best reply correspondence and dynamics}},
url = {http://doi.wiley.com/10.3982/TE652},
volume = {8},
year = {2013}
}

@article{Ritzberger1995,
author = {Ritzberger, Klaus and Weibull, Jorgen W.},
doi = {10.2307/2171774},
issn = {00129682},
journal = {Econometrica},
month = {nov},
number = {6},
pages = {1371},
title = {{Evolutionary Selection in Normal-Form Games}},
url = {https://www.jstor.org/stable/2171774?origin=crossref},
volume = {63},
year = {1995}
}

@article{fudenberg2014learning,
  title={Learning with recency bias},
  author={Fudenberg, Drew and Levine, David K and others},
  journal={Proceedings of the National Academy of Sciences},
  volume={111},
  pages={10826--10829},
  year={2014},
  publisher={Citeseer}
}

@book{young1998individual,
  author={Young, H. Peyton},
  year={1998},
  title={Individual Strategy and Social Structure An Evolutionary Theory of Institutions},
  publisher={Princeton University Press}
}

@article{Benaim2009,
author = {Bena{\"{i}}m, Michel and Hofbauer, Josef and Hopkins, Ed},
doi = {10.1016/j.jet.2008.09.003},
isbn = {0022-0531},
issn = {00220531},
journal = {Journal of Economic Theory},
number = {4},
pages = {1694--1709},
publisher = {Elsevier Inc.},
title = {{Learning in games with unstable equilibria}},
url = {http://dx.doi.org/10.1016/j.jet.2008.09.003},
volume = {144},
year = {2009}
}

@article{Basu1991,
author = {Basu, Kaushik and Weibull, J{\"{o}}rgen W.},
doi = {10.1016/0165-1765(91)90179-O},
issn = {01651765},
journal = {Economics Letters},
number = {2},
pages = {141--146},
title = {{Strategy Subsets Closed Under Rational Behavior}},
volume = {36},
year = {1991}
}

@phdthesis{Nash1950,
author = {Nash, John},
school = {Princeton University},
title = {{Non-cooperative games}},
type = {PhD thesis},
year = {1950}
}

@article{brown1951iterative,
  title={Iterative solution of games by fictitious play},
  author={Brown, George W},
  journal={Activity analysis of production and allocation},
  volume={13},
  number={1},
  pages={374--376},
  year={1951},
  publisher={New York}
}

@article{shapley1964some,
  title={Some topics in two-person games},
  author={Shapley, Lloyd},
  journal={Advances in game theory},
  volume={52},
  pages={1--29},
  year={1964}
}

@book{fudenberg1998theory,
  title={The theory of learning in games},
  author={Fudenberg, Drew and Drew, Fudenberg and Levine, David K and Levine, David K},
  volume={2},
  year={1998},
  publisher={MIT press}
}

@book{slotine1991applied,
  title={Applied nonlinear control},
  author={Slotine, Jean-Jacques E and Li, Weiping and others},
  volume={199},
  number={1},
  year={1991},
  publisher={Prentice hall Englewood Cliffs, NJ}
}

@article{benaim2003deterministic,
  title={Deterministic approximation of stochastic evolution in games},
  author={Bena{\"\i}m, Michel and Weibull, J{\"o}rgen W},
  journal={Econometrica},
  volume={71},
  number={3},
  pages={873--903},
  year={2003},
  publisher={Wiley Online Library}
}

@book{weibull1997evolutionary,
  title={Evolutionary game theory},
  author={Weibull, J{\"o}rgen W},
  year={1997},
  publisher={MIT press}
}

@book{sandholm2010population,
  title={Population games and evolutionary dynamics},
  author={Sandholm, William H},
  year={2010},
  publisher={MIT press}
}

@book{Camerer2003,
author = {Camerer, Colin F},
isbn = {0691090394},
publisher = {Princeton University Press},
title = {{Behavioral Game Theory: Experiments in Strategic Interaction}},
year = {2003}
}

@article{Ellison2000,
author = {Ellison, Glenn},
doi = {10.1111/1467-937X.00119},
issn = {0034-6527},
journal = {Review of Economic Studies},
number = {1},
pages = {17--45},
title = {{Basins of Attraction, Long-Run Stochastic Stability, and the Speed of Step-by-Step Evolution}},
url = {http://restud.oxfordjournals.org/lookup/doi/10.1111/1467-937X.00119},
volume = {67},
year = {2000}
}

\end{document}